\documentclass[twocolumn]{svjour3}          
\smartqed  
\usepackage{graphicx}
\usepackage{scrtime}
\usepackage[dont-mess-around]{fnpct}

\usepackage[misc,geometry]{ifsym}
\usepackage{times,amsfonts,graphicx}
\usepackage[perpage,symbol]{footmisc}
\usepackage[english]{babel}
\usepackage{commath}
\usepackage{pbox, booktabs, rotating}
\usepackage{array, makecell, lipsum}
\usepackage{threeparttable}
\usepackage{multirow,hyphenat,balance}
\usepackage{placeins}
\usepackage[compress]{cite}
\usepackage{float,color}
\usepackage[linesnumbered,ruled]{algorithm2e}
\usepackage{gensymb}
\usepackage{rotating, color} 
\usepackage{listings}
\usepackage{booktabs}
\usepackage[dvipsnames]{xcolor}
\usepackage{bm}
\usepackage{boldline}
\usepackage{xcolor}
\usepackage{balance, flushend}
\usepackage{amssymb}
\DeclareMathAlphabet{\mathcal}{OMS}{cmsy}{m}{n}

\usepackage{caption}
\usepackage[perpage,symbol]{footmisc}
\usepackage[english]{babel}
\usepackage{multirow,multicol}
\usepackage{amsmath, bm} 
\usepackage{pbox, booktabs, rotating}
\usepackage{subfigure}
\usepackage{array, makecell, lipsum}
\usepackage{enumerate}
\usepackage{threeparttable}
\usepackage{multirow,hyphenat,balance}
\usepackage{placeins}
\usepackage{float,color}
\usepackage{ragged2e}
\usepackage{enumerate}
\usepackage{etoolbox}
\usepackage{blkarray}
\usepackage[linesnumbered,ruled]{algorithm2e}
\let\oldnl\nl
\newcommand{\nonl}{\renewcommand{\nl}{\let\nl\oldnl}}

\usepackage{paralist}
\usepackage{cases}
\usepackage{url}
\usepackage{framed}
\usepackage{flushend}
\usepackage{esvect}
\usepackage[english]{babel}
\usepackage{kbordermatrix}
\usepackage{tikz} 
\usepackage{amssymb}  

\renewcommand{\kbldelim}{(}
\renewcommand{\kbrdelim}{)}

\newcolumntype{M}[1]{>{\centering\arraybackslash}m{#1}}

\AtBeginDocument{%
  \providecommand\BibTeX{{%
    \normalfont B\kern-0.5em{\scshape i\kern-0.25em b}\kern-0.8em\TeX}}}

\newcommand*{\affaddr}[1]{#1} 
\newcommand*{\affmark}[1][*]{\textsuperscript{#1}}
\newcommand*{\ie}{{\em i.e.}\@\xspace}
\newcommand*{\eg}{{\em e.g.}\@\xspace}
\newcommand*{\el}{{\em et al.}\@\xspace}
\makeatletter
\newcommand*{\etc}{%
    \@ifnextchar{.}%
        {etc}%
        {etc.\@\xspace}%
}

 \usepackage{mathptmx}      
\usepackage{latexsym}
 \journalname{J Electron Test}

\begin{document}

\title{AFIA: ATPG-Guided Fault Injection Attack on Secure Logic Locking}

\author{%
Yadi Zhong\affmark[1] \and Ayush Jain\protect\affmark[1] \and M. Tanjidur Rahman\affmark[2] \and Navid Asadizanjani\affmark[2] \and Jiafeng Xie\affmark[3] \and Ujjwal Guin\affmark[1]
      }
\institute{ Yadi Zhong \at
             \email{yadi@auburn.edu}
           \and
            Ayush Jain \at
              \email{ayush.jain@auburn.edu}
            \and  
            M. Tanjidur Rahman \at
             \email{mir.rahman@ufl.edu}
           \and
              Navid Asadizanjani \at
            \email{nasadi@ufl.edu}
            \and 
            Jiafeng Xie \at
              \email{jiafeng.xie@villanova.edu}
            \and
            \Letter \: Ujjwal Guin \at
              \email{ujjwal.guin@auburn.edu }           
              \\ \\
             \affaddr{\affmark[1]Department of Electrical and Computer Engineering, Auburn University, Auburn, AL, 36849, USA} \\
             \affaddr{\affmark[2]Department of Electrical and Computer Engineering, University of Florida, Gainesville, FL, 32611, USA}\\
            \affaddr{\affmark[3]Electrical and Computer Engineering, Villanova University, Villanova, PA, 19085, USA}
}

\date{Received: date / Accepted: date}

\maketitle

\begin{abstract}
The outsourcing of the design and manufacturing of integrated circuits has raised severe concerns about the piracy of Intellectual Properties and illegal overproduction. Logic locking has emerged as an obfuscation technique to protect outsourced chip designs, where the circuit netlist is locked and can only be functional once a secure key is programmed. However, Boolean Satisfiability-based attacks have shown to break logic locking, simultaneously motivating researchers to develop more secure countermeasures. In this paper, we present a novel fault injection-based attack to break any locking technique that relies on a stored secret key, and denote this attack as \textit{\textbf{AFIA}}, \textit{\textbf{A}}TPG-guided \textit{\textbf{F}}ault \textit{\textbf{I}}njection \textit{\textbf{A}}ttack. The proposed attack is based on sensitizing a key bit to the primary output while injecting faults at a few other key lines that block the propagation of the targeted key bit. AFIA is very effective in determining a key bit as there exists a stuck-at fault pattern that detects a stuck-at 1 (or stuck-at 0) fault at any key line. The average complexity of the number of injected faults for AFIA is linear with the key size $\mathcal{K}$ and requires only $\mathcal{K}$ test patterns to determine a secret key $K$. AFIA requires fewer injected faults to sensitize a bit to the primary output, compared to $2\mathcal{K}-1$ faults for the differential fault analysis attack~\cite{jain2020atpg}.
\end{abstract}


\maketitle

\section{Introduction} \label{sec:intro}
Over the last few decades, the impact of globalization has transformed the integrated circuit (IC) manufacturing and testing industry from vertical to horizontal integration. The continuous trend of device scaling has enabled the designer to incorporate more functionality in a system-on-chip~(SoC) by adopting lower technology nodes to increase performance and reduce the overall area and cost of a system. Currently, most SoC design companies or design houses no longer manufacture chips and maintain a foundry (fab) of their own. This is largely due to the increased complexity in the fabrication process as new technology development is being adopted. The cost for building and maintaining such foundries is estimated to be a multi-million dollar investment~\cite{YehFabCost2012}. As modern integrated circuits (ICs) are becoming more complex, parts of the design are reused instead of designing the whole from scratch. As a result, the design house integrates intellectual properties~(IP) obtained from different third-party IP vendors and outsources the manufacturing to an offshore foundry. Due to this distributed design and manufacturing flow, which includes designing SoCs using third-party IPs, manufacturing, testing, and distribution of chips, various threats have emerged in recent years~\cite{alkabani2007active, castillo2007ipp, tehranipoor2011introduction}. The research community has also been extensively involved in proposing countermeasures against these threats~\cite{roy2008epic, rajendran2012security, guin2016fortis, charbon1998hierarchical, kahng2001constraint, qu2007intellectual, jarvis2007split}.

\begin{figure}[t]
    \centering
    \includegraphics[width=\linewidth]{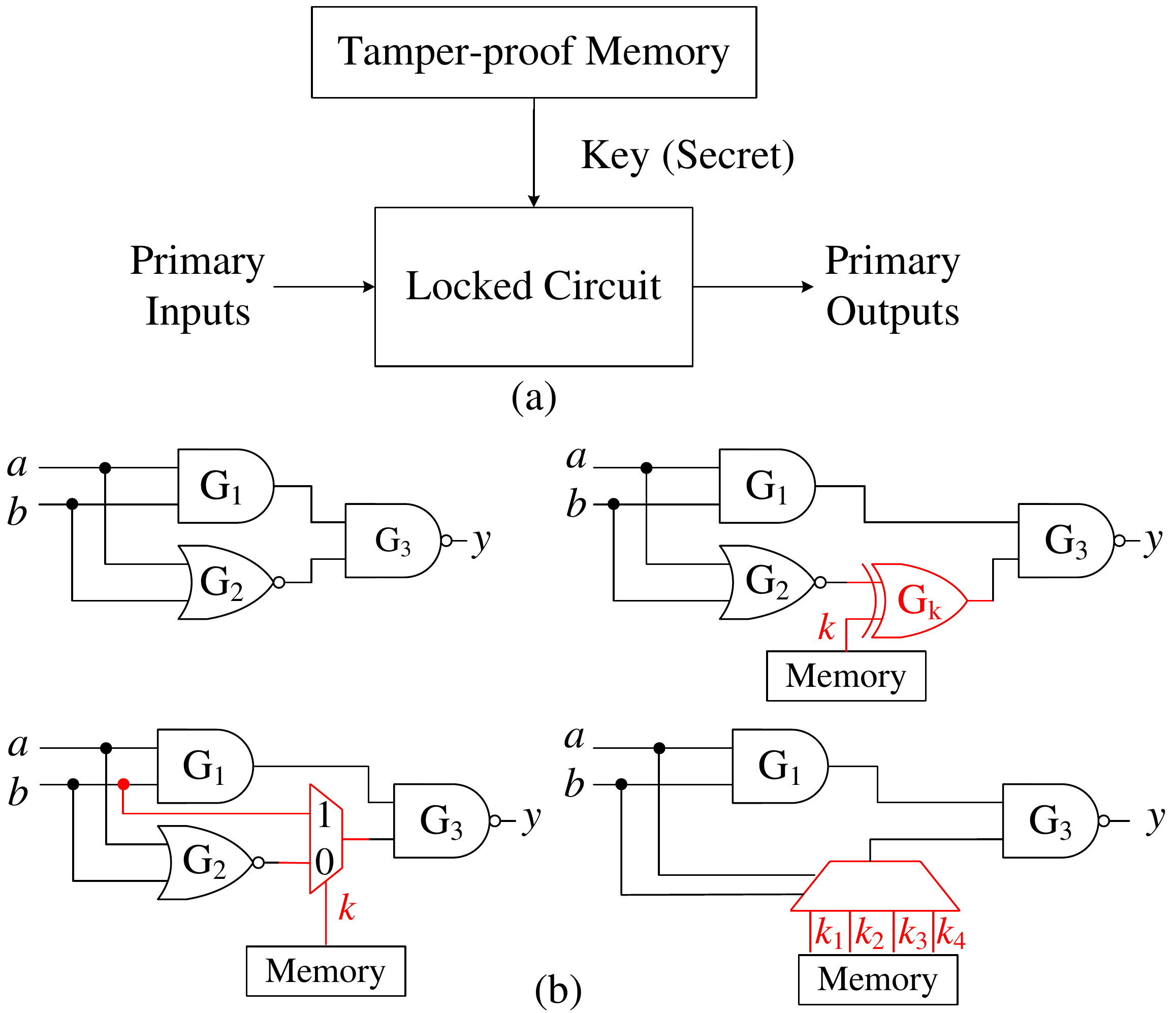} \vspace{-5px}
    \caption{Logic Locking: (a) An abstract view of the logic locking. (b) Different types of logic locking techniques with XOR/XNOR, MUX and LUT.} \vspace{-15px}
    \label{fig:LL}
\end{figure}

Logic locking has emerged as the most prominent method to address the threats from untrusted manufacturing~\cite{lee2006low, alkabani2007active, roy2008epic, chakraborty2008hardware}. In logic locking, the netlist of a circuit is locked with a secret key so that the circuit produces incorrect results in regular operation unless the same key is programmed into the chip. Figure~\ref{fig:LL}(a) shows an abstract view of logic locking where the key is stored in a tamper-proof memory and is applied to the locked circuit to unlock its functionality. The key needs to be kept secret, and care must be taken during the design process so that this secret key is not leaked to the primary output directly during the operation. The common logic locking techniques insert additional logic elements like XOR gates~\cite{roy2008epic}, multiplexers (MUXs)~\cite{rajendran2015fault}, and look-up tables (LUTs)~\cite{baumgarten2010preventing} to lock the circuit functionality, and are shown in Figure \ref{fig:LL}(b). SAT attack, by Subramanyan~\el~\cite{subramanyan2015evaluating}, was among the first ones to efficiently attack a range of locking schemes. With SAT analysis~\el~\cite{subramanyan2015evaluating}, the key of a locked circuit is determined in a short period of time. The SAT attack requires the locked netlist, recovered through reverse engineering, and a functional working chip. Since then, several SAT-resistant locking techniques have emerged~\cite{xie2016mitigating, wang2017secure, GuinTVLSI2018, yasin2017provably, sengupta2018customized, sweeney2020latch, azar2021data, rahman2022oclock, kamali2019full, roshanisefat2018srclock,kamali2018lut, shamsi2019icysat, xie2017delay} and many of them were broken soon after they have been proposed~\cite{shamsi2019ip, rahman2021security, duvalsaint2019characterization, duvalsaint2019characterization_seq, kamali2020interlock, limaye2019robust, el2017reverse, kamali2022advances, zhong2022complexity}. The majority of the research has been directed towards SAT attack resiliency. \textit{However, can we reliably state that a logic locking technique is completely secure even if we achieve complete SAT resistivity?} An untrusted foundry can be treated as an adversary as logic locking is proposed to protect designs from untrusted manufacturing. The adversary has many more effective means to determine the secret key without performing SAT analysis. A few of these attacks can be in the form of probing~\cite{rahman2020key, rahman2020defense}, inserting a hardware Trojan in the design~\cite{jain2021taal}, and analyzing the circuit topology~\cite{zhang2019tga, zhang2020novel, sirone2020functional}. Countermeasures are also developed to partially prevent these attacks~\cite{zhang2015veritrust, shen2018nanopyramid, vashistha2018trojan, zhang2019tga, zhang2020novel,sisejkovic2021deceptive, kamali2020interlock, alrahis2021unsail}.


Unlike cryptosystems, not all input patterns for a locked circuit are valid for propagating the incorrect key values to the primary outputs. Instead, only a few patterns may exist to carry the values of key bits to the output, similar to the identification of hard-to-detect faults. This is especially true for Post-SAT solutions~\cite{xie2019anti, yasin2017provably,yasin2016sarlock,sengupta2020truly,shakya2020cas}, where they minimize the output corruptibility for incorrect keys. For logic locking, some key bits can block the propagation of the target key bit, \ie SLL~\cite{rajendran2012security} and Post-SAT designs. This is different from the fault injection attack in cryptography, where an entirely new output can be observed under any input pattern even though there is a single bit change in the key as the plaintext goes through many transformations (e.g., Shift Rows, Mix Columns and key addition for AES)~\cite{AES2001, paar2009understanding, blomer2003fault}. It is trivial for a cryptosystem to change one key bit and apply a random pattern. Unfortunately, this is not the case for a circuit locked with a secret key. It is hard to observe the output change with the change of a single key bit by applying a random pattern. The novelty of this paper is that we apply the methodology in ATPG to efficiently derive the desired input pattern, which guarantees the change in output under different keys and helps launch the fault injection attack.

This paper shows how an adversary can extract the secret key from a locked netlist, even if all the existing countermeasures are in place. An adversary can determine the secret key by injecting faults at the key registers~\cite{rahman2020defense, rahman2020key}, which hold the key value during normal operation, and performing differential fault analysis. 
In this paper, we present \textit{\textbf{AFIA}}, key sensitization-based \textit{\textbf{A}}TPG-guided \textit{\textbf{F}}ault \textit{\textbf{I}}njection \textit{\textbf{A}}ttack, to break any locking scheme. The entire process can be performed in three steps. First, we process the locked netlist and converted it into a directed graph to extract all logic cones and construct a key-cone association matrix that records the distribution of keys among different cones. This structural analysis facilitates total fault reduction for subsequent test pattern generation. Second, it is necessary to select an input pattern that produces an incorrect response for the target key bit only while keeping its dependent keys at faulty states. This can be achieved by using a constrained automatic test pattern generation~(ATPG)~\cite{bushnell2004essentials} to generate such a test pattern, which is widely popular for testing VLSI circuits. It is a simple yet effective way to determine a 1-bit key by generating a test pattern that can detect the stuck-at fault (\textit{saf}) at the target key (corresponds to that key bit) while keeping its dependent keys at logic 1 (or logic 0). Dependencies are often inserted~\cite{yasin2015improving} to prevent direct sensitization of a key bit to the output by test patterns due to other key lines blocking its path. In our proposed approach, the pattern which detects a stuck-at 1 (\textit{sa1}) fault at one key line with logical constraints for the recovered key lines is sufficient to determine that key bit. One can also use stuck-at 0 (\textit{sa0}) fault to derive such pattern. 
Note that the fault-free and faulty responses are always the complements under the test pattern that detects that fault, which helps to derive the key bit value. The same process needs to be applied for other key bits to generate such input patterns, and this results in most {$\mathcal{K}$} patterns for determining the entire key of size {$\mathcal{K}$}. Note that one test pattern can detect multiple key bits when they are placed in different logic cones (no dependencies). Third, we apply these test patterns to only one instance of unlocked chip obtained from the market and collect the responses. Faults can be injected at the blocking key registers using laser fault injection equipment (see Section~\ref{subsec:laser-fault-injection} for details) and obtain the key value by comparing the output responses with test patterns' generated by constrained ATPG. This is a significant improvement compared to our previous conference paper~\cite{jain2020atpg} where differential fault analysis requires injection of faults twice.

The contributions of this paper are described as follows:

\begin{itemize}
    \item We propose a novel attack to break secure logic locking techniques using fault injection-based method. The basic idea behind the attack is the availability of an input pattern that sensitizes a key bit to the primary output. If there are interdependencies among keys, fault injection is necessary only for the dependent key bits in order to propagate the desired ones to the output. Multiple key bits can be sensitized to the outputs if they are placed in different logic cones during locking. \textit{\textbf{To the best of our knowledge, we are the first to demonstrate that the stuck-at fault patterns can be used to determine the secret key of a locked circuit with fault injections on interdependent keys.}} 
    
    \item The proposed attack can be launched very efficiently with the minimum number of injected faults. It is necessary to inject faults only to ensure the proper key propagation, whereas our prior work~\cite{jain2020atpg} requires $2\mathcal{K}-1$ faults ($\mathcal{K}-1$ faults for $C_A$ and $\mathcal{K}$ faults for $C_F$) to determine one key bit. In addition, our proposed cone analysis approach can find key bits which are located in different cones in parallel. As fault injection is an expensive process, we propose to generate test patterns that reduce the number of injected faults. Each key bit is targeted one at a time to minimize the number of faults. Note that fault injection is necessary when a group of key bits blocks the propagation of a targeted key bit to the output. 
    \item We demonstrate the feasibility of our proposed fault injection attack using Hamamatsu PHEMOS-1000, a laser fault injection equipment, on a Kintex-7 FPGA~\cite{xilinx}. We have performed extensive simulations on different benchmarks with secure locking techniques. Constrained ATPG using the Synopsys TetraMAX tool~\cite{SynopsysTetraMAX} is used to generate test patterns to simulate the attack. The simulation result shows a significant reduction of total fault count for AFIA compared to DFA~\cite{jain2020atpg} in breaking the same locked benchmark.
    
\end{itemize}

The rest of the paper is organized as follows: An overview of different logic locking techniques and existing attacks along with fault injection techniques is provided in Section~\ref{sec:prior-work}. We describe our previously published attack \cite{jain2020atpg} in Section \ref{sec:background}. The proposed attack and its methodology to extract the secret key from any locked circuit are described in Section~\ref{sec:AFIA}. We present the results for the implementation of the proposed attack on different locked benchmark circuits in section~\ref{sec:experimental-results}. Finally, we conclude our paper in Section~\ref{sec:conclusion}.

\section{Prior Work}\label{sec:prior-work} 
The prior work related to logic locking and fault injection techniques is described in this section.  

\vspace{-5px}
\subsection{Logic Locking}\label{logic-locking}
As mentioned in Section~\ref{sec:intro}, the objective of logic locking is to obfuscate the functionality of the original circuit by inserting a lock (secret key). The key-dependent circuit makes it difficult for the adversary to pirate or analyze the original circuit directly. In this context, various traditional logic locking techniques were based on different location selection algorithms for key gate placement, such as random~(RLL)~\cite{roy2010ending}, fault analysis-based~(FLL)~\cite{rajendran2012security}, and strong interference-based logic locking~(SLL)~\cite{rajendran2015fault}. To demonstrate the capabilities of an adversary, Subramanyan~\el~\cite{subramanyan2015evaluating} developed a technique using Boolean Satisfiability (SAT) analysis to obtain the secret key from a locked chip. This oracle-guided SAT attack iteratively rules out incorrect key values from the key space by using distinguishing input patterns and the corresponding oracle responses.

\begin{figure*}[t]
    \centering 
    \includegraphics[width=\linewidth]{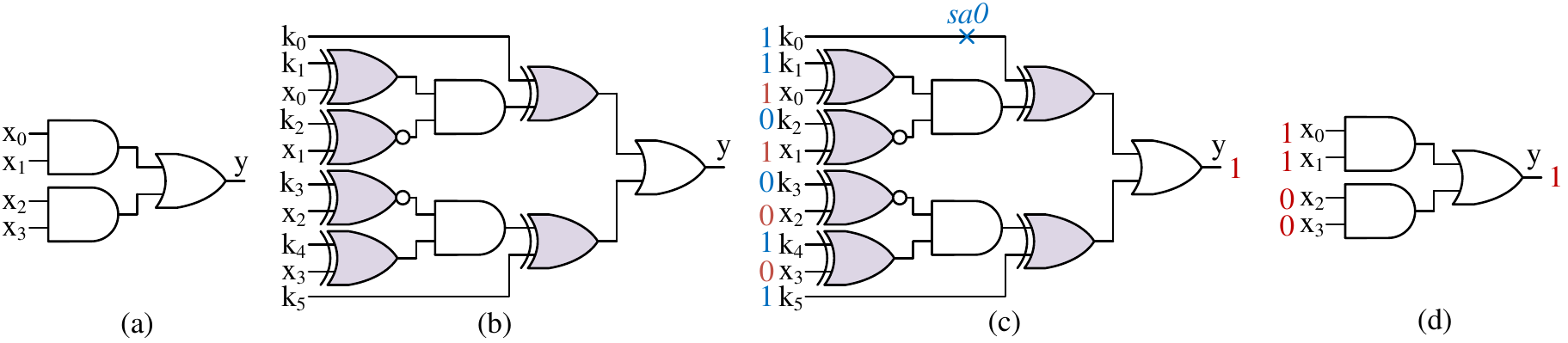}\vspace{-10px}
    \caption{The inefficiency of CLIC-A attack. (a) Original circuit. (b) Locked circuit with 6 dependent key gates, where correct key $\{k_0,...,k_5\}=\{\textit{001100}\}$. (c) Key assignment and input pattern returned by Constraint-based CLIC-A. (d) The oracle response $y=1$.} \label{fig:clic-a} 
\end{figure*}

In post-SAT era, resiliency against the SAT attack became one of the crucial metrics to demonstrate the effectiveness of newly proposed schemes~\cite{kamali2022advances}. Sengupta~\el proposed stuck-at-fault based stripping of original netlist and reconstruction to form the locked netlist, where incorrect results are produced only for chosen input patterns~\cite{sengupta2018customized, yasin2017provably}. Simultaneously, researchers have adopted a different direction to tackle the SAT attack, including restricting access to the internal states of a circuit through scan-chains. Guin \el proposed a design that prevents scanning out the internal states of a design after a chip is activated and the keys are programmed/stored in the circuit~\cite{guin2017novel, GuinTVLSI2018}. The concept of scan locking gained significant interest from the researchers, which led to the development of various scan-chain-based locking schemes~\cite{wang2017secure,karmakar2018encrypt, rahman2021security}. Alrahis \el attack scan-chain-based locking schemes by unrolling the sequential circuit to a combinational one, which is then provided to the SAT solver to extract the secret key~\cite{alrahis2019scansat2}. Sisejkovic \el~\cite{sisejkovic2021deceptive} proposed an oracle-less structural analysis attack on MUX-based (SAAM) logic locking to exploit insertion flaws in MUX-based key gates. Deceptive multiplexer-based (D-MUX) logic locking is proposed to achieve functional secrecy~\cite{shamsi2019impossibility} against both SAAM and oracle-less machine learning-based attacks. As combinational feedback loops are not translatable to SAT problems, cyclic-based locking~\cite{roshanisefat2018srclock} is resistant to the initial SAT attack~\cite{subramanyan2015evaluating}. In addition, there has been extensive efforts in the proposal of non-functional logic locking techniques, such as scan-chain-based~\cite{rahman2021security, azar2021cryptography}, timing-based locking~\cite{sweeney2020latch, azar2021data, rahman2022oclock, xie2017delay}, and routing-based locking~\cite{kamali2019full, kamali2018lut, kamali2020interlock}.

As the research community explores new directions to understand an attacker's latent qualities, new attacks on logic locking have been proposed. An adversary may perform direct or indirect probing on the key interconnects or registers~\cite{rahman2020key}. An attacker is not required to understand the complete functionality of the circuit to perform these attacks. In this, Rahman \el demonstrated how an attacker could target the key registers and perform optical probing to gain knowledge regarding the fixed value for those registers. Following this, tampering attacks can also become an attacker's primary choice. Jain~\el exploited this notion to extract the secret key by implementing hardware Trojans inside the locked netlist~\cite{jain2021taal}. Without an oracle, the attribute of repeated functionality in the circuit can also be used to compare the locked unit functions and their unlocked version to predict the secret key~\cite{zhang2019tga}. This makes it necessary to lock all instances of unit functions in the entire netlist to achieve a secured logic locking scheme. CLIC-A~\cite{duvalsaint2019characterization, duvalsaint2019characterization_seq}, an ATPG-based attack, can break keys by applying constraint-based ATPG to propagate the target key bit to the primary output but suffers scalability in the dependent key count. Cyclic-based locking has suffered from modified SAT-based attack~\cite{shamsi2019icysat}, where cyclic-based constraints are placed to avoid infinite loops. Several non-functional-based locking can be broken by sequential-based attacks with limited scan access~\cite{kamali2022advances, limaye2019robust, el2017reverse} or SMT attack~\cite{azar2019smt}.

\subsubsection{Comparison of AFIA with CLIC-A~\cite{duvalsaint2019characterization}}

There is a major difference between our proposed \textit{AFIA} and \textit{CLIC-A}. First, the worst-time complexity for test generation regarding the total test pattern count for solving the key-dependent faults between \textit{AFIA} and \textit{CLIC-A} differs significantly. Our worst-case complexity of solving an $n$-bit key of non-mutable convergent key gates inside a single logic cone is at most $n$ test patterns with $\frac{n\cdot (n-1)}{2}$ injected faults (see Theorem~\ref{th:worst-case}). This is because AFIA determines each key bit by directly comparing the output response with the generated test pattern. On the other hand, \textit{CLIC-A} applies constraint-based ATPG by assigning constraint on the $(n-1)$-bit key and setting a stuck-at 0 at the target key line since placing don't cares (\textit{X}) on other key bits will not produce the desired test patterns for non-mutable convergent key gates. However, the simulated output from the constraint-based ATPG likely agrees with the oracle simulation under the same test pattern hardness of logic locking. Note that it does not mean that the constraints placed on the $(n-1)$-bit key is the correct key values, and it only indicates that, under the particular test pattern, the output from the netlist with constraints and the stuck-at fault matches with the oracle output. Instead, \textit{CLIC-A} has to perform additional constraints in ATPG and check the output against the oracle to ensure that key values are correct. The worst-case complexity in the total test pattern count for \textit{CLIC-A} is exponential $\bm{O}(2^{(n-1)})$. 

Let us consider an example of an unlocked circuit in Figure~\ref{fig:clic-a}(a) locked with 6 dependent XOR/XNOR key gates, as shown in Figure~\ref{fig:clic-a}, whose correct key $\{k_0,...,k_5\}=\{\textit{001100}\}$. As none of the keys can be sensitized to the output without knowing the correct value for the other five, CLIC-A runs constraint-based ATPG and sets \textit{sa0} to $k_0$. Suppose ATPG returns a test pattern with key vector $\{k_0,...,k_5\}=\{\textit{110011}\}$ and input vector $\{x_0,...,x_3\}=\{\textit{1100}\}$, along with the simulated fault-free output $\mathtt{atpg\text{-}sim}(x_0,...,x_3,k_0,...,k_5)=1$, as shown in Figure~\ref{fig:clic-a}(c). Although the output of the simulated netlist matches with the oracle response $y=1$, the key value returned by constraint-based ATPG is incorrect. There is no method for CLIC-A to check whether the key vector is the actual key other than appending it as a constraint to ATPG so that the test pattern returned at the next iteration would be different from the current one. The worst-case complexity for CLIC-A to fully determine the 6-bit key is to iterate through all possible combinations of the remaining 5-bit key (excluding $k_0$ with \textit{sa0}), resulting in a $2^{5}$ test pattern count to break the locking scheme with dependent keys. On the other hand, AFIA only requires $6$ test patterns to determine all 6-bit keys, which is much more efficient than CLIC-A. In summary, test pattern generation for CLIC-A becomes infeasible if there are a large number of dependent keys in a logic cone.

\subsubsection{Comparison of AFIA with Key Sensitization Attack~\cite{rajendran2012security}}

There is a similarity between our proposed AFIA and sensitization attack~\cite{rajendran2012security}. The similarity between these approaches is the sensitization, i.e., the propagation, of the key to the output. However, our approach is more general for the following reason. First, the sensitization attack does not need fault activation as the key gates are XOR/XNOR gates, and the key can propagate to the key gate output for both input 0 and 1. However, this may not hold for non-XOR-based locking techniques. For example, MUX-based locking has keys connected to the input of AND gate instead of the XOR gate, where one needs to set the other AND input to the non-controlling value 1 for fault activation. Besides, it is common practice for recent locking techniques to synthesize the locked benchmark after key insertion. The synthesis tool can optimize the key gate with other gate types, which results in keys directly connected to non-XOR gates like AOI, NAND, etc. To propagate the key value to the primary output, having only key sensitization without the activation would not work for synthesized locked circuits. For example, we can break \textit{SFLL-hd}~\cite{yasin2017provably}, \textit{SFLL-flex}~\cite{yasin2017provably}, and \textit{SFLL-rem}~\cite{github-sfllrem} with $n$ patterns for a $n$-bit key (see Section~\ref{subsec:attack-analysis}), where sensitization attack requires brute force attack ($\bm{O}(2^n)$) to all the non-mutable keys in the \textit{SFLL} restoration circuitry. Second, our proposed fault injection can break non-mutable convergent key gates from strong logic locking, which is the countermeasure proposed in a sensitization attack. AFIA only needs at most $n$ (see Theorem~\ref{th:pattern-count}) test patterns for a $n$-bit pairwise non-mutable convergent keys, but it would take $\bm{O}(2^{n})$ in the worst case to brute force the correct key under sensitization attack~\cite{rajendran2012security}.

\subsubsection{Dissimilarities between Logic Locking and Cryptosystems} \label{sec:comparsion-crypto}
There has been considerable efforts~\cite{yasin2017provably,beerel2022towards} in the proposal of formal analysis on logic locking through introducing similar concepts used in cryptography. However, logic locking techniques differ from various cryptosystems in two aspects. First, the objective for logic locking and cryptosystem is different. The cryptographic algorithm ensures that the secret key is fully integrated with the input plaintext (\ie, the addRoundKey in all ten rounds of AES encryption). Logic locking, however, focuses on perturbing the output, commonly by XORing a 1-bit key with a wire in the circuit,  under certain input patterns, where no repeated insertion of the same key bit or its derived value to elsewhere. Second, the output of a locked circuit and the ciphertext of a cryptosystem behaves differently under input combinations. A locked circuit under an incorrect key may behave identically to the oracle (or locked circuit with the correct key) under multiple input patterns. This is particularly true for Post-SAT locking solutions, \ie, SARLock~\cite{yasin2016sarlock}, Anti-SAT~\cite{xie2019anti}, SFLL~\cite{yasin2017provably,sengupta2020truly}, CAS-Lock\cite{shakya2020cas}, where the output corruptibility for incorrect keys is reduced to the bare minimum. This means that a locked circuit with an incorrect key behaves exactly as an unlocked circuit under an exponential number of input combinations. 

The cryptographic algorithms, especially the block ciphers, are built on confusion and diffusion properties recommended by Claude Shannon in his classic 1949 paper~\cite{shannon1949communication}. This results in a large number of output bit changes in the output (ciphertext) even for a single bit change in the input (plaintext)~\cite{AES2001,paar2009understanding}. For example, AES has 10/12/14 rounds of diffusion and confusion operations depending on the key size of 128/192/256 bits. It is thus trivial to launch differential fault analysis as it will guarantee the change in the output, where one can compare the faulty and fault-free responses by injecting a fault into a key register, one at a time. On the contrary, digital circuits generally do not have repeated layers of operations like block ciphers. Digital circuits, except crypto accelerators, are designed to meet the user specification of speed, power, and area, and the functionality (change in output) depends on the user's needs. It is well understood and verified that digital circuits have lots of don't cares ($Xs$) in the inputs. The VLSI test community adopted test compression~\cite{rajski2004embedded, SynopsysTestCompression} to reduce the test pins and resultant test times. As there exists a large number of $Xs$ in the test pattern, it is infeasible to apply a random pattern and expect it to propagate the target key bit (e.g., a stuck-at fault at the key line) to output. For example, if there are 70\% Xs in a test pattern with a 100 input cone [which is very common], the probability of a random pattern propagating the key to the output is $2^{30}/2^{100}\approx 0$. The effect of some keys in a locked circuit can even be muted due to the circuit's structural and functional behavior~\cite{rajendran2012security}, which is in direct contrast to cryptosystems, where every output is influenced by all key bits~\cite{paar2009understanding}.

\subsection{Fault Injection Methods}\label{subsec:fault-techniques}
Over the years, several threats and methods have emerged to break a cryptosystem without performing mathematical analysis or brute force attacks. Using these attacks, an adversary can subvert the security of protection schemes, primarily through extracting or estimating the secret key using physical attacks. Fault injection attacks intentionally disturb the computation of cryptosystems in order to induce errors in the output response. To achieve this, external fault injection is performed through invasive or non-invasive techniques. This is followed by the exploitation of erroneous output to extract information from the device.

Fault-based analysis on cryptosystems was first presented theoretically by Boneh \el on RSA~\cite{boneh1997importance}. This contribution initiated a new research direction to study the effect of fault attacks on cryptographic devices. The comparison between the correct and faulty encryption results has been demonstrated as an effective attack to obtain information regarding the secret key~\cite{piret2003differential,dusart2003differential,lee2019high}. These can be realized into different categories:

\vspace{2px}
\noindent $\tiny \bullet$ \textit{Clock Glitch}: The devices under attack are supplied with an altered clock signal which contains a shorter clock pulse than the normal operating clock pulse. For successfully inducing a fault, these clock glitches applied are much shorter than the circuit's tolerable variation limit for the clock pulse. This results in setup time violations in the circuit and skipping instructions from the correct order of execution~\cite{fukunaga2009practical, selmane2008practical}.

\vspace{2px}
\noindent $\tiny \bullet$ \textit{Power Variation}: This technique can be further bifurcated into two subcategories: either the malicious entity may choose to provide a low power supply to the system (also abbreviated as underfeeding), or the adversary may choose to influence the power line with spikes. This adversely affects the set-up time and influences the normal execution of operations. The state elements in the circuit are triggered without the input reaching any stable value, causing a state transition to skip operations or altering the sequence of execution~\cite{guilley2008silicon, barenghi2013fault,barenghi2010low}.

\vspace{2px}
\noindent $\tiny \bullet$ \textit{Electromagnetic Pulses/Radiation}: The eddy current generated by an active coil can be used to precisely inject faults at a specific location in the chip. This method does not require the chip to be decapsulated in order to inject the fault. However, the adversary is required to possess information regarding specific modules and their location inside the chip~\cite{schmidt2007optical, dehbaoui2012electromagnetic}.

\vspace{2px}
\noindent $\tiny \bullet$ \textit{Laser}: Fault injection using lasers is also regarded as a very efficient method because it can precisely induce a fault at an individual register to change its value~\cite{barenghi2012fault}. For optical fault injection, the laser can be focused on a specific region of the chip from the backside or front side. However, due to the metal layers on the front side, it is preferred to perform the attack on the backside of the chip. Skorobogatov \el~\cite{skorobogatov2002optical} first demonstrated the effectiveness of this method by using a flashgun to inject fault to flip a bit in the SRAM cell. Several other research groups also utilized and proposed different variants of this method to study the security of cryptographic primitives~\cite{skorobogatov2010optical, canivet2011glitch, pouget2007tools, selmke2016attack}.

\vspace{2px}
\noindent $\tiny \bullet$ \textit{Focused-ion Beam~(FIB)}: The most effective and expensive fault injection technique is devised with focused ion beam~(FIB)~\cite{torrance2009state}. This method enables cutting/connecting wires and even operates through various layers of the IC fabricated in the latest technology nodes~\cite{wu2014precise}. 

\vspace{2px}
\noindent $\tiny \bullet$ \textit{Software Implemented Fault Injection}:  This technique produces errors through software that would have been produced when a fault targeted the hardware. It involves the modification of programs running on the target system to provide the ability to perform the fault injection. It does not require dedicated complex hardware, a gate-level netlist, or RTL models that are described in hardware description languages. The faults are injected into accessible memory cells such as registers and memories through software that represent the most sensitive zones of the chip~\cite{skarin2010goofi,tsai1995ftape,hsueh1997fault}.

\section{Background} \label{sec:background}
In this section, we present a differential fault analysis (DFA) attack introduced in~\cite{jain2020atpg}. Our attack method is inspired by VLSI test pattern generation. One test pattern is able to detect a single stuck-at fault with the propagation of this fault to the primary output. Since key values from tamper-proof non-volatile memory are loaded to key registers, these registers are the potential locations for stuck-at-faults. With an active chip at hand, the adversary could target these registers and extract the secret key.  

\subsection{Threat Model}
The threat model defines the capabilities of an adversary and its standing in the IC manufacturing and supply chain. It is very important to know an attacker's ability and the available resources/tools to estimate its potential to launch the attack. The design house or entity designing the chip is assumed to be trusted. The attacker is assumed to be the untrusted foundry or a reverse engineer having access to the following:
\begin{itemize}
    \item The locked netlist of a circuit. An untrusted foundry has access to all the layout information, which can be extracted from the GDSII or OASIS file. Also, this locked netlist can be reconstructed by reverse engineering the fabricated chip in a layer-by-layer manner with advanced technological tools~\cite{torrance2009state}.
    \item An unlocked and fully functional chip is accessible to the adversary since the chip is publicly available from the market.  
    \item A fault injection equipment is essential to launch the attack. It is not mandatory to use high-end fault injection equipment. The main operation is to inject faults at the locations of key registers (all the flip-flops) on a de-packaged/packaged chip. Precise control is not necessary as we target all the flip-flops simultaneously. An adversary can also choose the software methods to inject faults at these flip-flops. Once the register is at the faulty state, the scan enable (\textit{SE}) signal needs to be assigned to put the chip in test mode. 
    \item The attacker has the know-how to determine the location of the tamper-proof memory. Then, it will be trivial for an adversary to find the location of the key register in a netlist, as it can easily trace the route from the tamper-proof memory.
\end{itemize}

\noindent\textbf{Notations}: To maintain uniformity across the entire paper, we represent frequently used terms with the defined notations, and they will be referred to with these notations in the following subsections.

\begin{itemize} 
    \item $\mathcal{K}$ denotes key length or key size, i.e., the number of bits in the key.
    \item $K$ denotes the keyspace; $K=\{k_0, k_1, \ldots k_{\mathcal{K}-1}\}$.
    \item The locked netlist of a circuit is abbreviated as $C_L$. The unlocked and fully functional chip/circuit, whose tamper-proof memory has been programmed with the correct key, is denoted by $C_{O}$. The two versions of fault-injected circuits are described as follows:
    \begin{itemize}
        \item $C_{F}$ represents a locked circuit where all the key lines ($\mathcal{K}$) are injected with logic 1 (or logic 0) faults. We call it the circuit with faulty key registers for differential fault analysis (DFA).
        \item $C_{A}$ represents the same locked circuit in which $(\mathcal{K}-1)$ key lines are injected with the same logic 1 (or logic 0) faults, leaving one key line fault-free. We denote this circuit as a fault-free circuit for DFA.  
    \end{itemize}
    
    For any given circuit, we assume the primary inputs~($PI$) of size \textit{$|PI|$}, primary outputs~($PO$) of size \textit{$|PO|$}, and secret key~($K$) size of \textit{$\mathcal{K}$}. We also use key lines or key registers alternatively throughout this paper as their effects are the same on a circuit.
    \vspace{5px}
    \item Stuck-at fault (\textit{saf}): For any circuit modeled as a combination of Boolean gates, stuck-at fault is defined by permanently setting an interconnect to either 1 or 0 in order to generate a test vector to propagate the fault value at the output. Each connecting line can have two types of faults, namely, stuck-at-0 (\textit{sa0}) and stuck-at-1 (\textit{sa1}). Stuck-at faults can be present at the input or output of any logic gates~\cite{bushnell2004essentials}.
    \vspace{5px}
    \item Injected fault: A fault is injected at the key register using a fault injection method (see details in Section~\ref{sec:prior-work}).   
\end{itemize}

Note that \textit{saf} is an abstract representation of a defect to generate test patterns, whereas an injected fault is the manifestation of a faulty logic state due to fault injection.

\begin{figure}[ht]
    \centering
    \includegraphics[width=\linewidth]{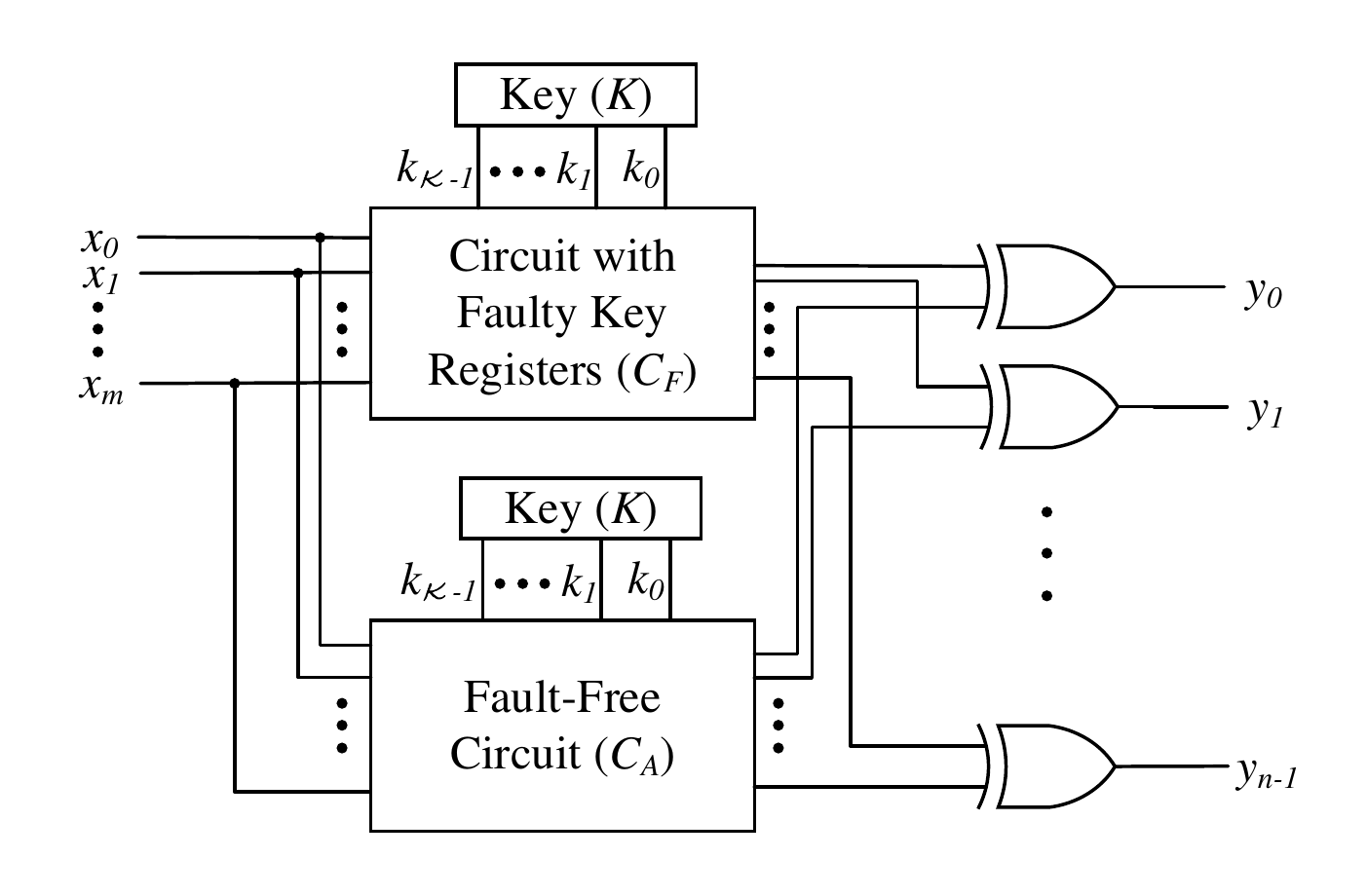}
    \caption{The abstract representation of our DFA attack.} 
    \label{fig:absract-FIAL}\vspace{-10px}
\end{figure}

\subsection{Differential Fault Analysis (DFA) Attack Methodology}\label{subsec:dfa}
This fault injection attack relies on differential fault analysis. The captured output response of the circuit with faulty key registers with the corresponding fault-free circuits can reveal the key. Applying any fault injection methods (see the details in Section~\ref{subsec:fault-techniques}), the attacker can create the faulty chip/circuit. Figure~\ref{fig:absract-FIAL} shows an abstract representation of DFA. The fault-free circuit ($C_A$) is an unlocked chip ($C_{O}$) bought from the market whose key bits need to be retrieved. Except for the key-bit targeted to be extracted, all remaining key registers are fixed to a particular faulty value of either 0 or 1 corresponding to the selected fault. A circuit with faulty key registers ($C_F$) uses the same chip, and it is injected with a particular fault to keep all the key registers or interconnects to a faulty value of logic 1 or 0. One input pattern is first applied to $C_A$, and its response is collected. The same input pattern is then applied to the $C_F$ to collect the faulty response. By XORing the corresponding circuit response, any output discrepancy between fault-free circuit ($C_{A}$) and the circuit with faulty key registers ($C_F$) is revealed. If both the circuits differ in their responses, the XORed output will be 1; otherwise, it will be 0. If we find an input pattern that produces a conflicting result for both $C_A$ and $C_F$ only for one key bit, the key value can be predicted. The key value is the same as the injected fault value if the XORed output is of logic 0; otherwise, the key value is a complement to the injected fault. 

The attack can be described as follows: 
\begin{itemize}
    \item \textit{Step-1}: The first step is to select an input pattern that produces complementary results for the fault-free ($C_A$) and faulty ($C_F$) circuits. The input pattern needs to satisfy the following property --
    it must sensitize only one key bit to the primary output(s). In other words, only the response of one key bit is visible at the \textit{PO}, keeping all other key bits at logic 1s (or 0s). If this property is not satisfied, it will be impractical to reach a conclusion regarding the value of a key bit. \textit{Now the question is, how can we find if such a pattern exists in the entire input space ($\xi$)}. 
    
    \hspace{10px} To meet this requirement, our method relies on stuck-at faults~(\textit{saf}) based constrained ATPG to obtain the specific input test patterns (see details in Section~\ref{subsec:test-pattern-gen}). Considering the fact that the adversary has access to the locked netlist, it can generate test patterns to detect \textit{sa1} or \textit{sa0} at any key lines and add constraints to other key lines~(logic 1 and 0 for \textit{sa1} and \textit{sa0}, respectively). A single fault, either \textit{sa0} or \textit{sa1} on a key line, is sufficient to determine the value of that key bit. Therefore, we have selected \textit{sa1}, and the following subsections are explained considering this fault only. This process is iterated over all the key bits to obtain {$\mathcal{K}$} test patterns. The algorithm to generate the complete test pattern set is provided in Algorithm~Section~\ref{subsec:test-pattern-gen}.
    
    \vspace{5px}
    \item \textit{Step-2}: The complete set of generated test patterns is applied to the fault-induced functional circuit with faulty key registers~($C_F$). The circuit is obtained by injecting logic 1 fault on the key registers if \textit{sa1} is selected in the previous step; else, the circuit is injected with logic 0 faults for \textit{sa0}. The responses are collected for later comparison with fault-free responses. For $C_A$, test patterns are applied such that it matches the fault modifications in the circuit. For example, the test pattern for the first key is applied to the circuit when the circuit instance does not pertain to any fault on its corresponding key register and holds the correct key value while the remaining key registers are set to logic 1 (for \textit{sa1}) or 0 (for \textit{sa0}). For the next key-bit, ($C_A$) instance is created by excluding this selected key bit from any fault while keeping all other key registers to logic 1 (for \textit{sa1}) or 0 (for \textit{sa0}). This process is repeated for all key bits, and their responses are collected for comparison in the subsequent step.

    \vspace{5px}
    \item \textit{Step-3}: The adversary will make the decision regarding the key value from the observed differences in the output responses of ($C_A$) and ($C_F$). For any test pattern corresponding to a particular key bit, when the outputs from both circuits are the same, it implies that the injected fault on the key lines in a $C_F$ circuit is the same as the correct key bit; only then will the outputs of both ICs be same. Otherwise, when $C_F$ and $C_A$ differ in their output response, it concludes the correct key bit is a complement to the induced fault. This process is repeated for all key bits. In this manner, the key value can be extracted by comparing the output responses of both circuits for the same primary input pattern.
    
\end{itemize}

\subsection{Example} \label{subsec:examples}
We choose a combinational circuit as an example for simplicity to demonstrate the attack. The attack is valid for sequential circuits, as well, as it can be transformed into a combinational circuit in the scan mode, where all the internal flip-flops can be reached directly through the scan-chains~\cite{bushnell2004essentials}.    

\begin{figure}[ht]
    \centering
    \includegraphics[width=\linewidth]{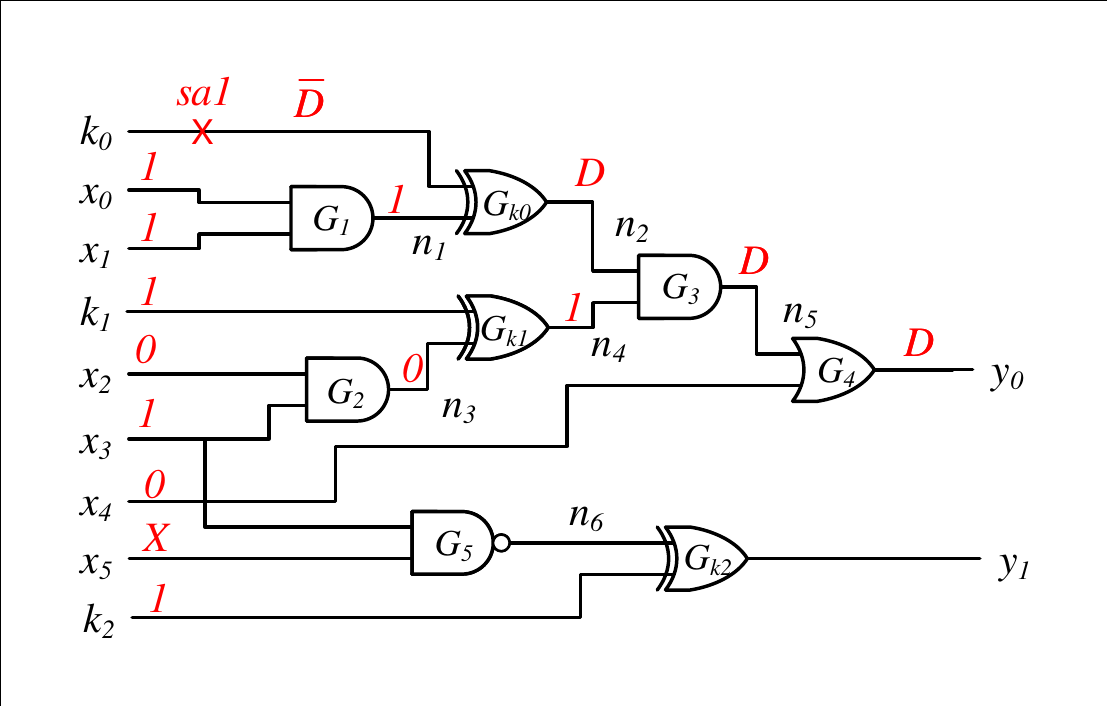}\vspace{-10px}
    \caption{Test pattern generation considering a \textit{sa1} at key line $k_0$ with constraint $k_1=1$ and $k_2=1$. Test pattern, $P_1=[11010X]$ can detect a \textit{sa1} at $k_0$. }
    \label{fig:fault-propagation}
\end{figure}

Figure~\ref{fig:fault-propagation} shows the test pattern generation on a circuit locked with a 3-bit secret key, where the propagation of $k_0$ is dependent on $k_1$ and vice versa. First, we target to find out the value of $k_0$. A test pattern $P_1$ is generated to detect a \textit{sa1} fault at $k_0$ with constraint $k_1=1$ and $k_2=1$ (adding faults on all the key lines except the target key bit). As the value of $k_1$ is known during the pattern generation, the effect of the \textit{sa1} at $k_0$ will be propagated to the primary output $y_0$. For a fault value $\overline{D}$ at $k_0$, if $[x_0~x_1] = [1~1]$ then $D$ propagates to $n_2$. To propagate the value at $n_2$ to the output of $G_3$, its other input ($n_4$) needs to attain logic 1. Since $k_1=1$ due to injected fault which is set as a constraint in ATPG tool, $n_4=1$ for $n_3=0$ which implies $[x_2~x_3] = [0~1]$. At last, $x_4=0$ propagates $D$ propagates the value at $n_5$ to the primary output $y_0$. The output $y_0$ can be observed as $D$ for the test pattern $P_1=[1~1~0~1~0~X]$. Finally, to perform the DFA, this pattern $P_1$ needs to be applied to both $C_A$ and $C_F$ to determine the value of $k_0$. Similar analysis can be performed for the other two key bits, $k_1$ and $k_2$.

\subsection{Test Pattern Generation} \label{subsec:test-pattern-gen}
To generate the test pattern set, an automated process relying on constrained ATPG is performed. The detailed steps to be followed are provided in Algorithm~\ref{alg:TP-generation}. Synopsys Design Compiler~\cite{SynopsysDC} is utilized to generate the technology-dependent gate level netlist and its test protocol from the RTL design. A test protocol is required for specifying signals and initialization requirements associated with design rule checking in Synopsys TetraMAX~\cite{SynopsysTetraMAX}. Automatic test generation tool TetraMAX generates the test patterns for the respective faults along with constraints for the locked gate level netlist.   

\setlength{\textfloatsep}{5pt}
\begin{algorithm}[t]
\SetAlgoLined
\SetKwInOut{Input}{Input}\SetKwInOut{Output}{Output}
\Input{~Locked gate-level netlist ($C_L$), test protocol ($T$), and standard cell library}
\Output{~Test pattern (\textit{P}) set}
\vspace{-5px}
\nonl \rule{0.45\textwidth}{0.4pt}

Read the locked netlist ($C_L$) \;
Read standard cell library \;

Run design rule check with test protocol generated from design compiler \;

Determine key size~{$\mathcal{K}$} from $C_L$ \; 

\For{$i\gets0$ \KwTo ($\mathcal{K}-1$) }{
    Add a \textit{sa1} fault at key line $k_i$ \;
        
        \For{$j\gets0$ \KwTo ($\mathcal{K}-1$)}{
            \If{$i \neq j$}{
                Add constraint at $k_j$ to logic 1 \;
            }
        }
        Run ATPG to detect the fault \;
        Add the test pattern, $P_i$ to the pattern set, \textit{P} \;
        Remove all faults \;
        Remove all constraints \;
    }
Report the test pattern set, \textit{P} \;
\caption{Test pattern generation for constrained ATPG in DFA} \label{alg:TP-generation}
\end{algorithm} 

The inputs to the algorithm are the locked gate-level netlist~($C_L$), Design Compiler generated test protocol~($T$), and the standard cell library. The algorithm starts with reading the locked netlist and standard cell library~(Lines 1-2). The ATPG tool runs the design rule check with the test protocol obtained from the Design Compiler to check for any violation (Line 3). Only upon the completion of this step is the fault model environment set up in the tool. The size of the key ($\mathcal{K}$) is determined by analyzing $C_L$ (Line 4). The remaining key lines are selected one by one to generate test patterns (Line 5). A stuck-at-1 fault is added at the $i^{th}$ key line to generate $P_i$ (Line 6). The ATPG constraints (logic 1) are added to other key lines (Lines 7-11). A test pattern $P_i$ is generated to detect the \textit{sa1} at the $i^{th}$ key line (Lines 12-13) and added to the pattern set, $P$. All the added constraints and faults are removed to generate the $(i+1)^{th}$ test pattern (Lines 14-15). Finally, the algorithm reports all the test patterns, $P$ (Line 17).

\vspace{-5px}
\section{AFIA: ATPG-guided Fault Injection Attack} \label{sec:AFIA}

The objective of an adversary is to reduce the number of injected faults to launch an efficient attack. The DFA presented in Section~\ref{subsec:dfa} requires $2\mathcal{K}-1$ faults to determine a single key bit, where $\mathcal{K}$ denotes the secret key size. This severely limits the adversary's capability as injecting a large number of faults is challenging from the fault injection equipment's perspective. All these faults need to be injected when applying the test pattern to evaluate one key bit. In this section, we present an efficient attack and denoted as \textit{\textbf{AFIA}}, an \textit{\textbf{A}}TPG-guided \textit{\textbf{F}}ault \textit{\textbf{I}}njection \textit{\textbf{A}}ttack based on key sensitization. This new attack only requires injecting the fault on a key register if there is a dependency among keys. The threat model remains the same as DFA. We consider an untrusted foundry to have access to the gate-level netlist and can generate manufacturing test patterns.

\begin{figure}[t]
    \centering
    \includegraphics[width=\linewidth]{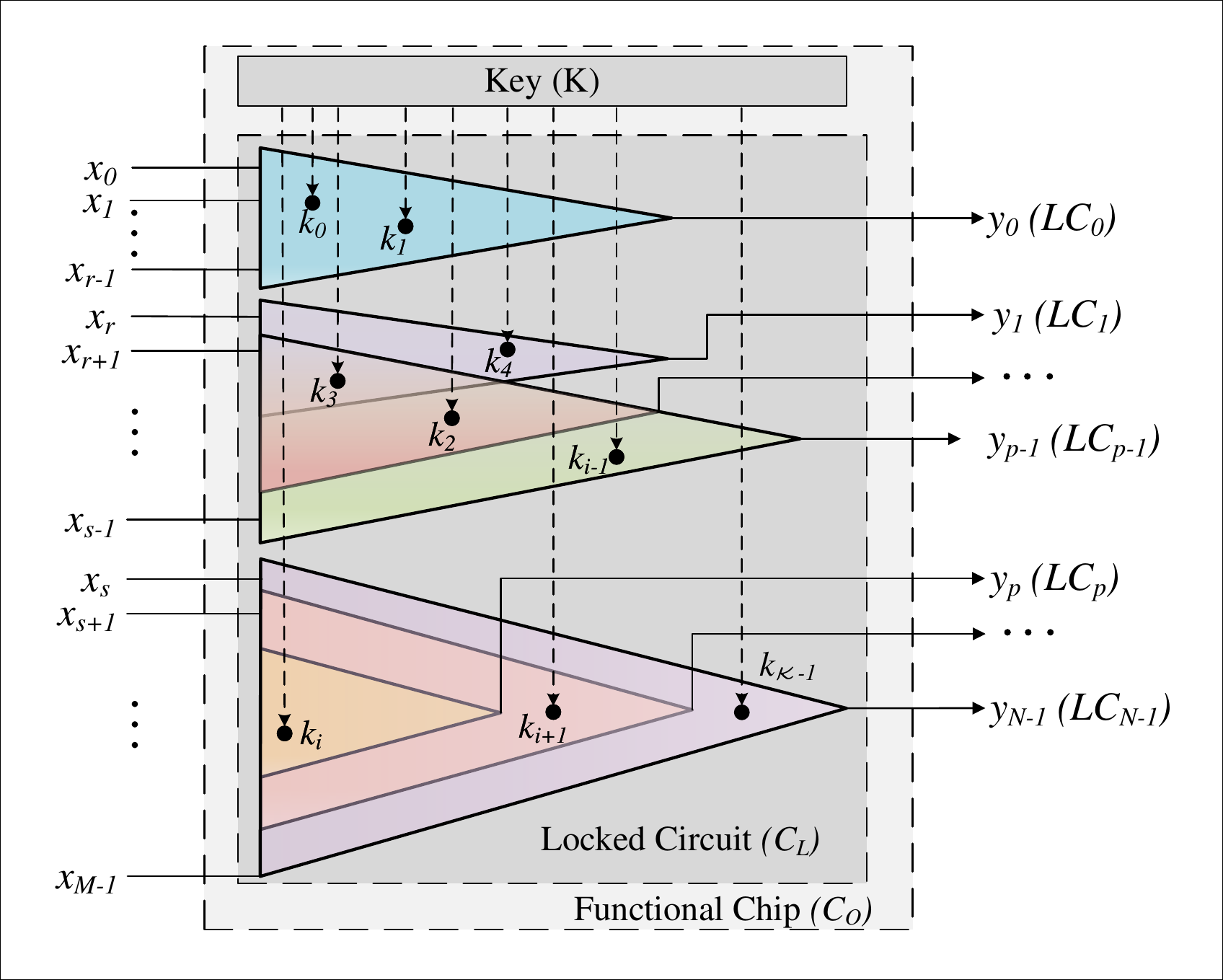} 
    \caption{An abstract view of a locked circuit.}
    \label{fig:FI} 
\end{figure}

\subsection{Overall Approach} The proposed attack AFIA evaluates one key bit at a time iteratively and can be summarized by the following steps:

\begin{itemize}
    \item \textit{Step-1}: First, AFIA analyzes the locked circuit $C_L$ and its logic cones. Some cones are completely independent (e.g., $LC_0$ in Figure \ref{fig:FI}), some cones share few inputs (e.g., $LC_1$, \ldots, $LC_{p-1}$), and the others share the same inputs (e.g., $LC_p$, \ldots, $LC_{N-1}$). It is necessary to determine keys from cones that are a subset of other larger cones (if any) first during the test pattern generation in order to reduce the number of injected faults. For an independent logic cone (say $LC_0$), we can propagate the keys one at a time without injecting faults at keys of other cones. If the two cones are overlapped, it is beneficial to sensitize keys to a cone with fewer unknown keys. 

    \item \textit{Step-2}: Similar to DFA, it requires an input pattern to derive a correct key bit. We denote this key bit as the target key bit. Constraints are set on the recovered key lines, where no fault injections are needed. The attacker performs fault injection (\textit{Step-3}) solely on keys (in the same cone) that block the propagation of the targeted key bit. The blocking key set is determined by the returned test patterns from ATPG TetraMAX~\cite{SynopsysTetraMAX}. Once a key bit is determined, AFIA targets the next key bit of the same cone by putting the previously obtained keys as constraints during the test pattern generation.

    \item \textit{Step-3}: The last step applies fault injections on functional chip $C_O$ using the generated test patterns of \textit{Step-2}. The targeted key value can be extracted by comparing the fault-injected output against the output pattern computed by ATPG. When the value of all the targeted key bits in one text pattern has been identified, we can constrain these bits with their actual values in ATPG in the subsequent pattern.
\end{itemize}

AFIA is an iterative method, where \textit{Step-2} is performed to generate test patterns, and \textit{Step-3} injects fault and applies that pattern to determine the targeted key bit. Once this targeted key is determined, it will be used as a constraint in \textit{Step-2}. The following subsections present these three steps in detail.

\subsection{Cone Analysis}\label{subsec:coneA}
The goal of this proposed attack is to apply minimal fault injections to recover the complete key set. It is ideal for the adversary to inject faults at key registers only when necessary. In general, not all keys prevent the propagation of the target key bit, as many of the keys are often distributed across the netlist and reside in different logic cones. A logic cone is a part of the combinational logic of a digital circuit that represents a Boolean function and is generally bordered by an output and multiple inputs~\cite{bushnell2004essentials}. Thus, cone analysis can effectively separate the dependence of different groups of key bits, where one group does not block the propagation of the key bits in other groups. We propose to analyze the internal structure of the locked netlist $C_L$ by creating a directed graph $G$ from it. We denote that both the inputs and logic gates' outputs are nodes. A directed edge exists from Node $n_1$ to Node $n_2$ if and only if they are associated with a logic gate. Intuitively, a circuit with $N$ outputs has $N$ logic cones, as in Figure~\ref{fig:FI}. Note that the number of cones can be only primary outputs (POs) for a combinational circuit or the sum of POs and pseudo primary outputs (PPOs) for a sequential circuit~\cite{bushnell2004essentials}. All the inputs and logic gates whose logical values affect $y_j$ belong to logic cone $LC_j$. The graph representation of logic cone $LC_j$ with sink $y_j$ is a subgraph of $G$. 

Two possible scenarios might occur during the locking of a netlist. Key bit(s) can be placed uniquely in a logic cone and cannot be sensitized to any other POs/PPOs except the cone's output. Other key bits can be placed in the intersection of multiple cones and can be sensitized through any of these. We observe that the majority of the key bits are inside the intersections with multiple cones. What should be the best strategy to propagate a key bit to one of the POs/PPOs when there exist multiple sensitization paths? Our objective is to reduce the number of faults to sensitize a key bit to a PO/PPO, and it is beneficial to select a cone with the minimum number of keys. Note that the keys in a cone can block the propagation of a targeted key in that same cone {only} and requires fault injection to set a specific value to these blocking keys. It is, thus, necessary to construct a key-cone association matrix \textit{A} to capture the correlation between the logic cones and the key bits. The matrix $A$ not only provides insight on which keys (and how many of them) are inside a logic cone but also offers a structured view of whether a key belongs to multiple logic cones, and is presented as follows: 
\begin{eqnarray*}
A &=& [a_{i,j}]_{\mathcal{K}\times N} \\
  &=&\hspace{-10pt} \kbordermatrix{
    & LC_0 & LC_1 & \dots & LC_{N-1}\\ 
    k_0 & a_{0,0} & a_{0,1} & \dots & a_{0,N-1}\\
    k_1 & a_{1,0} & a_{1,1} & \dots & a_{1,N-1}\\
    k_2 & a_{2,0} & a_{2,1} & \dots & a_{2,N-1}\\
    \vdots & \vdots & \vdots & \ddots & \vdots\\
    k_{\mathcal{K}-1} & ~a_{\mathcal{K}-1,0} & a_{\mathcal{K}-1,1} & \dots & a_{\mathcal{K}-1,N-1}\\ 
  },
\end{eqnarray*}

\noindent where, $a_{i,j} \in \{0, 1\}$, and $a_{i,j}=1$ if key $k_i$ is present in cone $LC_j$, otherwise, $a_{i,j}=0$.

It is straightforward for the attacker that, if he/she picks cone $LC_j$ and key bit $k_i$ (if its value is still unknown) in this cone, only keys (other than $k_i$) residing in $LC_j$ could potentially impede the propagation of $k_i$ to the output $y_j$. This is advantageous to the attacker because the keys outside of cone $LC_j$ would not, by any means, affect the propagation of  $k_i$ to $y_j$. Thus, he/she can safely ignore these keys, and it does not matter whether he/she already has the correct logical values for them or not.

\begin{figure}[t]
\centering
\includegraphics[width=\linewidth]{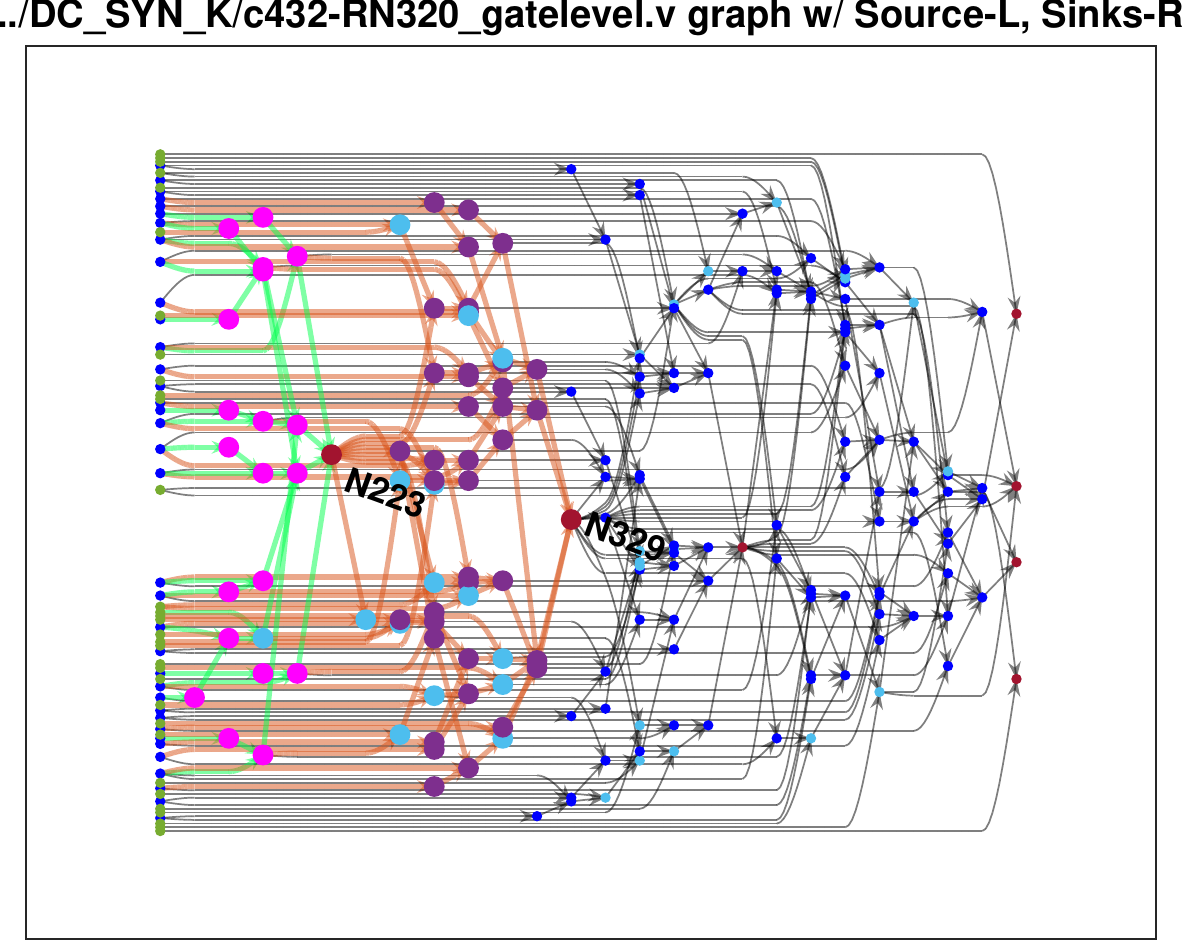}
\caption{Directed graph of locked c432-RN320 netlist with a 32-bit key.} 
\label{fig:c432-RN320}
\end{figure}

For example, the directed graph representation of locked netlist c432-RN320 with a 32-bit key~\cite{salmani2018trust} is shown in Figure \ref{fig:c432-RN320}. Output nodes are in red, key registers in green (at the left-most level), key gates in cyan, remaining input (at the left-most level), and gates in blue. The top two logic cones with the fewest keys are $LC_{N223}$ of output $N223$ and $LC_{N329}$ of output $N329$. Logic cone $LC_{N223}$ has only one key (\textit{keyIn\_0\_4}, with key gate highlighted) (all other nodes and edges are in magenta and light green). Logic cone $LC_{N329}$ is the superset of $LC_{N223}$, and it contains additional thirteen keys (all other nodes and edges exclusively in $LC_{N329}$ are in purple and orange). With AFIA, the only key in $LC_{N223}$ is determined first, followed by the remaining thirteen keys in $LC_{N329}$. Because of the only key in $LC_{N223}$, no fault injection is necessary for this key’s propagation to $N223$.

\subsection{Test Pattern Generation}\label{subsec:tpg}
Once the cone analysis is performed, it is required to generate test patterns so that a targeted key can be sensitized to one of the PO/PPO. The test pattern generation process is similar to the DFA presented in Section~\ref{subsec:dfa} except with a much lesser number of ATPG constraints. We treat undetermined keys as inputs during the test pattern generation and the recovered keys as ATPG constraints. As the secret key remains the same in an unlocked chip, it is unnecessary to inject faults at the recovered key bits as their values are known during the test pattern generation. On the other hand, we need to inject faults at unknown and yet to be determined key lines. However, it is not necessary to inject faults at all of them. We use the ATPG tool to determine whether one or more unknown key bits do not block the propagation of the targeted key bit. As we treat unknown keys as inputs, the ATPG tool can generate a pattern that might contain $X's$ at some of the key lines (using \textit{set\_atpg -fill X}~\cite{SynopsysTetraMAX}), and we do not need to inject faults at these bits. This allows an adversary to reduce the number of fault injections further. Similar to DFA, a stuck-at fault, \textit{sa1} (or \textit{sa0}), is placed on the target key bit with constraints on recovered key bits during the ATPG. When TetraMAX~\cite{SynopsysTetraMAX} returns a test pattern, the attacker applies the pattern and injects faults (presented in Section~\ref{subsec:fi}) to sensitize the target key bit at the PO/PPO. After recovering one key bit, AFIA sets ATPG constraints on the recovered key lines, generates another test pattern, and applies it to sensitize the next key. 

\subsection{Fault Injection} \label{subsec:fi}
The final step applies fault injections on functional chip $C_O$ using generated test patterns from Section \ref{subsec:tpg}. Faults are injected at the key registers with any appropriate fault injection techniques described in Section~\ref{subsec:fault-techniques}. No fault injection is necessary at the key bits whose values are already determined as their values are no different from those already programmed in the chip $C_O$. If we receive a faulty response by applying the test pattern developed in \textit{Step-2}, the value of the secret key will be 1 as we have sensitized a \textit{sa1} fault during the ATPG; otherwise, the secret key is 0. If we generate a test pattern considering a \textit{sa0} fault, the faulty response results in the secret key of 0, and vice versa. \textit{Step-2} in Section \ref{subsec:tpg} and \textit{Step-3} in Section \ref{subsec:fi} are repeated until the entire secret key is found. Consequently, fewer faults are injected compared with the DFA since injections happen only at key locations (of the same logic cone) that block the propagation of the to-be-determined key bits.

\subsection{Proposed Algorithm for AFIA} \label{subsec:alg-afia}


\begin{algorithm}[t]
\setlength{\intextsep}{1\baselineskip}
\SetAlgoLined
\SetKwInOut{Input}{Input}\SetKwInOut{Output}{Output}\SetKw{Continue}{continue}
\Input{~Locked gate-level netlist ($C_L$)}
\Output{~Secret key (\textit{KEY})}

\vspace{-5px}
\nonl \rule{0.45\textwidth}{0.4pt}

// ---------------------- Cone Analysis -------------------------------- \\
$[K, Y, G]$ $\gets$ \texttt{netlist2Graph(}$C_L$\texttt{)}\; 

\textit{Gflip} $\gets$ \texttt{flipEdges(}\textit{G}\texttt{)}\;
$A \gets [\ ]$ \;
\For {\bf{each} $y_j$ \textbf{in} \textmd{Y} }{
$[LK_j$, $LC_j]\gets$ 
\texttt{extractCone(}\textit{Gflip}, $K$, $y_j$\texttt{)}\;
$A \gets$ append vector $LK_j$ as the last column \;
}
\nonl // ------------- ATPG test pattern generation ------------------------\\
Recovered key bits from \textit{Step-3} of AFIA, $K^R \gets \varnothing$ \;


\While{($A\ !\hspace{-2pt}=\textbf{false}$)}{
$[K^U_{LC}] \gets$  \texttt{fConeWMinKeys(}$A,K$\texttt{)} \;
\If{$K^U_{LC} != \varnothing$}{ 
\For{$l\gets0$ \KwTo ($|K^U_{LC}|-1$) }{
Add a \textit{sa1} fault at key line $K^U_{LC}[l]$ \;
Add constraints at recovered key bits to\hspace{-1pt} \textit{K}$^R$\;
Test pattern $P_l\gets$ run ATPG (\textit{set\_atpg -fill X})\;
Remove all faults \;
Remove all constraints \;
\nonl // ---------------- Fault Injection ------------------------\\
Invoke \textit{Step-3} of AFIA with $P_l$\;
Add recovered key $K^U_{LC}[l]$ to $K^R$ \; 
Assign \textbf{false} to all entries in key $K^U_{LC}[l]$'s row in \textit{A} \;
}
}
}
Report the secret key, \textit{KEY} $\gets \{K,K^R\}$ \;
\caption{AFIA: ATPG-guided Fault Injection Attack.} \label{alg:afia}
\end{algorithm}

Algorithm~\ref{alg:afia} describes the implementation details of AFIA. The adversary first constructs a directed graph $G$ from the locked netlist $C_L$ (Line 1), as elaborated in Section \ref{subsec:coneA}. Aside from converting netlist to graph, \texttt{netlist2Graph(.)} returns the key list $K$ and output list $Y$. By exploiting directed graph structure, logic cone $LC_j$ can be easily extracted by flipping all edges in graph $G$ (Line 2) and run breadth-first-search (BFS) or depth-first-search, (DFS)~\cite{cormen2009introduction}, on output nodes $y_j$. The key-cone association matrix \textit{A} is declared as an empty array, where the cone and key information will be added (Line 3). 
Function \texttt{extractCone(.)} is implemented with BFS. It returns the directed subgraph of logic cone $LC_j$ and a logical (\textbf{true}/\textbf{false}) vector $LK_j$ of dimension $\mathcal{K} \times 1$. If key bit $k_q$ is inside cone $LC_j$, $LK_j[q]=$ \textbf{true}; else, $LK_j[q]=$ \textbf{false}. Matrix $A$ is updated by concatenating all vectors $LK_j$'s together (Line 6) so that the complete $A$ has $\mathcal{K}$ rows and N columns, as explained in Section \ref{subsec:coneA}. 

AFIA invokes \texttt{fConeWMinKeys(.)} (Line10) and obtains a vector $K^U_{LC}$ of all unknown keys in the logic cone with the fewest (positive) unknown keys. For simplicity, $K^U_{LC}$ records the row indices of the unknown keys, as in matrix $A$. For every key bit in $K^U_{LC}$, the \textit{sa1} is set on the to-be-determined key (Line 13). The recovered key values in $K^R$ are appended as constraints (Line 14). Test pattern $P_l$ (Line 15) is generated after invoking ATPG. All the stuck-at faults (Line 16) and constraints (Line 17) are removed. When $P_l$ and fault injections (Line 18) are applied on the working chip $C_O$, $K^U_{LC}[l]$ bit is recovered by referencing the ATPG's predicted output of the corresponding $P_l$. Afterward, the correct bit value is added to the recovered key list $K^R$ (Line 19). Since this bit is recovered, it is no longer an unknown key, and AFIA updates the association matrix \textit{A} to assign logical zero to all entries on key $K^U_{LC}[l]$’s row (Line 20). This is conceptually equivalent to deleting $K^U_{LC}[l]$ from the unknown key list as \texttt{fConeWMinKeys(.)} will only count the number of non-zero entries per column. When all key bits in $K^U_{LC}$ are determined, the adversary moves on to the subsequent logic cone (Line 10). Finally, when all cones are covered, the secret key \textit{KEY} is returned (Line 24).

\subsection{Example}
Here, we use the same circuit as in Figure \ref{fig:fault-propagation} as an example to illustrate how AFIA works. The circuit has six inputs, two outputs, and three key bits. With two outputs, this circuit has two logic cones, as in Figure \ref{fig:Key_Sen_example}. 
The same D-Algorithm~\cite{bushnell2004essentials} is applied to show the propagation of stuck-at-faults. Based on cone analysis in Section \ref{subsec:coneA}, logic cone $LC_0$ contains two key bits, $k_0$, $k_1$, cone $LC_1$ has only one key $k_2$. Thus, the association matrix \textit{A} can be represented as:
\[
  A=\kbordermatrix{
    & LC_0 & LC_1 \\ 
    k_0 & 1 & 0 \\
    k_1 & 1 & 0 \\
    k_2 & 0 & 1 \\
  }.
\]

\begin{figure*}[t!]
    \centering
    \includegraphics[width=\linewidth]{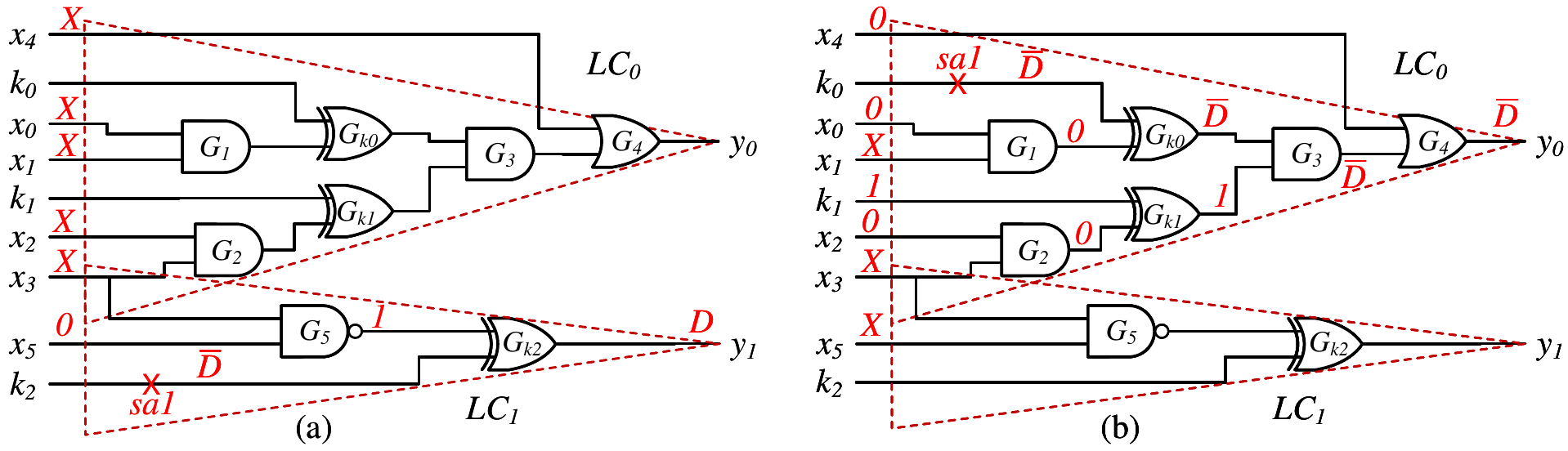}
    \caption{Test Pattern Generations for AFIA. (a) Test Pattern $P_0=[XXXXX0]$ for \textit{sa1} at $k_2$. (b) Test Pattern $P_1=[0X0X0X]$ for \textit{sa1} at $k_0$ with injected fault $k_1=1$.}
    \label{fig:Key_Sen_example}
\end{figure*}

AFIA picks a logic cone with the fewest number of unknown keys to solve (Line 10, Algorithm \ref{alg:afia}). Since all keys are unknown at this time, \texttt{fConeWMinKeys(.)} function selects logic cone $LC_1$ and returns $K^U_{LC}=[2]$. This cone has one key bit $k_2$, to which we assign \textit{sa1}. Using D-Algorithm, fault value $\overline{D}$ is marked on this key line. Here, the output $y_1$ is directly connected to XOR key gate $G_{k_2}$, and we can propagate this fault $\overline{D}$ to output $y_1=D$ with logic 1 for the other input of this XOR gate, as in Figure \ref{fig:Key_Sen_example}(a). Test pattern $P_0=[x_0x_1\ldots x_5]=[$\textit{XXXXX0}$]$ can detect \textit{sa1} for key $k_2$. Here, the value of the recovered key is 1 when the output is faulty. Otherwise, the recovered key is 0 as we have sensitized a \textit{sa1} fault during the ATPG. Note that no fault injection is necessary to determine this key. Matrix \textit{A} is updated with all zeros on the $k_2$'s row,
\begin{align*}
  A=\kbordermatrix{
    & LC_0 & LC_1 \\ 
    k_0 & 1 & 0 \\
    k_1 & 1 & 0 \\
    k_2 & 0 & 0 \\
  }.
\end{align*}
In the next iteration (Line 10), there is only one logic cone (also the cone with the least unknown keys), $LC_0$, left in matrix $A$ that has unknown keys. Function \texttt{fConeWMinKeys(.)} identifies $LC_0$ and yields $K^U_{LC}=[0\ 1]^T$, which captured the indices of unknown keys $k_0$, $k_1$. 
With two keys $k_0$ and $k_1$, AFIA chooses $k_0$ first randomly (Algorithm \ref{alg:afia} Line 13). By adding \textit{sa1} at $k_0$, test pattern $P_1=[x_0x_1\ldots x_5]=[$\textit{0X0X0X}$]$ with logic \textit{1} fault on $k_1$ can propagate the faulty response $\overline D$ in $k_0$ to $y_0$, as shown in Figure \ref{fig:Key_Sen_example}(b).  
Fault injection is performed at $k_1$ by setting its value to 1, and apply $P_1$ to determine $k_0$. AFIA, then, flushes out all the entries on row $k_0$ of matrix \textit{A},
\[
  A=\kbordermatrix{
    & LC_0 & LC_1 \\ 
    k_0 & 0 & 0 \\
    k_1 & 1 & 0 \\
    k_2 & 0 & 0 \\
  }.
\]

After $k_0$ is recovered, AFIA moves on to determining the other key in $LC_0$, $k_1$, (Line 12). We add a \textit{sa1} at $k_1$ (Line 13), along with constraining on $k_0$, $k_2$ to their determined values (Line 14). If the correct logical value for $k_0$ is \textit{0} (i.e., the stored key), test pattern $P_2=[x_0x_1\ldots x_5]=[$\textit{110X0X}$]$ can sensitize the \textit{sa1} of $k_1$ to the output $y_0$. If the stored secret key bit is $k_0=1$, the test pattern $P_2$ will be different, and its value will be $[$\textit{0X0X0X}$]$, which one can verify using the same D-Algorithm. Note that no fault injection is necessary to determine $k_1$. 


Finally, the matrix \textit{A} will be updated to all zeros 
and the AFIA recovers the entire key.

\subsection{AFIA Complexity Analysis} \label{subsec:attack-analysis}

The average complexity of the AFIA attack is linear with the key size ($K$). In this section, we show that AFIA is very effective at breaking any logic locking technique. However, the fault injection time may vary depending on the effectiveness of the equipment. It is practically instantaneous to obtain the secret key once the responses are collected from $C_O$.

\begin{lemma} \label{lemma:key-recovery}
    One input pattern is sufficient to recover one key bit. 
\end{lemma}

\begin{proof}
    A single test pattern is sufficient to detect a \textit{saf} if such a fault is not redundant~\cite{bushnell2004essentials}. A redundant fault results from a redundant logic that cannot be exercised from the inputs. As the key gates are placed to modify the functionality, it cannot be a redundant logic. As there exists one test pattern to detect a \textit{saf} at the key line, it can be used to recover one key bit. 
\end{proof}

\begin{theorem} \label{th:pattern-count}
AFIA recovers the entire secret key, $K$ using at most $\mathcal{K}$ number of test patterns, \ie,

\begin{equation}
TP_{AFIA}[f_K(C_{L}) = f(C_{O})] \leq \mathcal{K}.
\end{equation}
where $f_K()$ represents the functionality with $K$ as the key.
\end{theorem}
 
\begin{proof}
A $C_L$ with a $\mathcal{K}$-bit key is injected with a \textit{saf} fault on every key line. As AFIA requires one test pattern to obtain one key bit (see Lemma~\ref{lemma:key-recovery}), the upper bound of the number of test patterns is $\mathcal{K}$. However, a single pattern can detect two or more stuck-at faults on the key lines if their effect is visible in different logic cones (\eg, different outputs). As a result, the required number of test patterns to recover the entire key ($K$) can be less than $\mathcal{K}$.    
\end{proof}

\begin{theorem}
AFIA is applicable to strong logic locking~\cite{rajendran2012security}, where pairwise key gates are inserted to block the propagation of one key by the other.
\end{theorem}
 
\begin{proof}
In strong logic locking, the propagation of one key is blocked due to the other key. However, $(\mathcal{K}-1)$ faults are injected at $(\mathcal{K}-1)$ key lines, worst-case scenario, except for the one whose value needs to be determined. Once an external fault is injected into the functional chip, the key value is fixed and no longer remains unknown. Hence, AFIA is applicable to strong logic locking.     
\end{proof}

\begin{theorem} \label{th:worst-case}
The worst-case complexity for the total number of faults injected in AFIA is $\bm{O}(\mathcal{K}^2)$.
\end{theorem}
\begin{proof}
Let us consider a circuit with a single logic cone locked with a secret key vector $\{k_0, \ldots, k_{\mathcal{K}-1}\}$. Suppose all key bits are pairwise non-mutable convergent, i.e., the propagation of one key bit depends on all the other keys. To sensitize the $1^{st}$ key bit, we need to add $\mathcal{K}-1$ faults during the fault injection process. The $2^{nd}$ key bit requires $\mathcal{K}-2$ faults as the value of the $1^{st}$ key bit is known. Similarly, the $3^{rd}$ key bit requires $\mathcal{K}-3$ faults, and so on. Thus, the total number of faults is: \vspace{5px}

~~~~~~~~~~~~~~~~~~~~~~$\sum_{i=1}^{\mathcal{K}} (\mathcal{K}-i)=\frac{\mathcal{K}\cdot(\mathcal{K}-1)}{2}$.

Thus, the worst-case complexity for the total number of faults injected is $\bm{O}(\mathcal{K}^2)$.
\end{proof}

\begin{theorem}\label{th:average-worst}
The average-case complexity for the total number of faults injected in AFIA is $\bm{O}(\mathcal{K})$. 
\end{theorem}

\begin{proof}
Consider a circuit with $N$ logic cones, each cone $LC_j$ has negligible or no overlap with its neighboring cones, $LC_{j-1}$ and $LC_{j+1}$, and $\mathcal{K}$ keys are evenly distributed (amortized) among the $N$ cones. For each cone, it has an average $a=\frac{\mathcal{K}}{N}$ keys). Since negligible overlap between cones, there is no preference between the order of execution on deciphering keys in logic cones, and each cone needs to inject $\frac{\mathcal{K}/N\cdot(\mathcal{K}/N-1)}{2}$ faults. Overall, by summing up all faults for every logic cone, the required number of fault injections is $N\cdot\frac{\mathcal{K}/N\cdot(\mathcal{K}/N-1)}{2}$. 

Thus, the average-case complexity is $N\cdot\frac{\mathcal{K}/N\cdot(\mathcal{K}/N-1)}{2}= \frac{a-1}{2}\cdot \mathcal{K}= \bm{O}(a\mathcal{K})=\bm{O}(\mathcal{K})$.
\end{proof}

\subsection{AFIA on Fault-Tolerant Circuit}
Fault-tolerant circuits and circuits with redundancy may prevent the injected faults from being revealed at the output. However, it does not affect our proposed AFIA. As the objective of logic locking is to produce incorrect output for wrong key combinations under certain input patterns, these input patterns ensure the differential output behavior for keys. Thus, the key cannot be inserted inside the region of redundancy, where no input pattern can ever produce differential output. Any key bit placed at these locations cannot corrupt the output so that either logic 0 or logic 1 is its correct value. The SoC designer would not place a key bit in such a way that both logic values gives the correct output since it contradicts the principle of logic locking. In summary, redundancies are not a countermeasure against AFIA attack for a well-designed locked circuit.

\subsection{AFIA on Non-Functional-Based Locking Techniques}
Our fault injection-based attack can also be extended to non-functional logic locking techniques~\cite{rahman2021security, kamali2019full}. The dynamically obfuscated scan-chain (DOSC) technique~\cite{rahman2021security} has three secrets stored in the tamper-proof memory, which are the functional obfuscation key, the LFSR seed, and the control vector. AFIA can break the functional obfuscation key if the obfuscated scan-chain becomes transparent to the attacker. To achieve that, the attacker needs to inject faults at all the \textit{Scan Obfuscation Key} registers directly to get a known shift out state from the functional IP. For the routing-based locking technique~\cite{kamali2019full}, our proposed attack is applicable to breaking the key-configurable logarithmic-based network (CLN) as the switch-boxes (SwB) consist of MUX-based key gates. Once a fault is injected into a key register, the selection path for the corresponding MUX is determined. We can target these keys one at a time with test patterns generated from the ATPG tool and inject faults on dependent key registers.

\vspace{10px}
\section{Experimental Results} \label{sec:experimental-results} 

This section provides the feasibility of fault injection to break secure logic locking. Extensive simulations are performed on different benchmarks with different locking techniques to demonstrate the effectiveness of the proposed fault injection attack for breaking a secure locking technique. We have shown a significant reduction of total fault count for AFIA compared to DFA, presented in our conference paper, in breaking the same locked benchmark.

\begin{figure}[t]
        \centering
        \includegraphics[width = \linewidth]{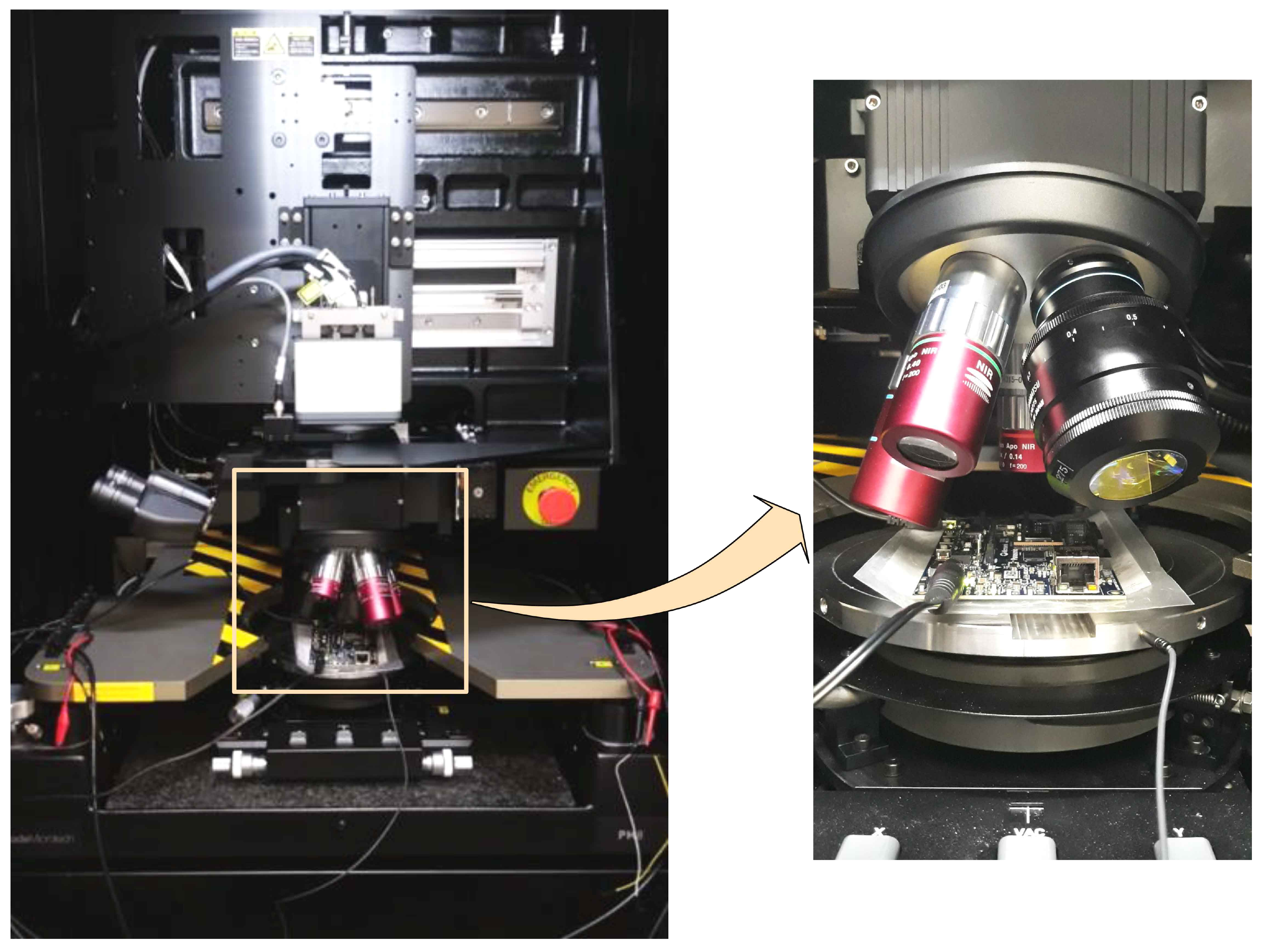}
        \caption{The FPGA board placed under the lens for laser-fault injection at the target registers.}
        \label{fig:setup}
\end{figure}

\vspace{-15px}
\subsection{Laser Fault Injection}\label{subsec:laser-fault-injection}\vspace{-5px}
To demonstrate the laser fault injection attack, we selected a Kintex-7 FPGA~\cite{xilinx}, which is used as the device-under-test (DUT). Locked benchmark circuits are implemented in the Kintex-7 FPGA, where faults are injected into key registers. Figure~\ref{fig:setup} shows the laser fault injection~(LFI) setup with a Hamamatsu PHEMOS-1000 FA microscope~\cite{PHEMOS-1000}. The equipment consists of a diode pulse laser source (Hamamatsu C9215-06) with a wavelength of 1064 nm. Three objective lenses were used during this work: 5x/0:14 NA, 20x/0:4 NA, 50x/0:76 NA. The 50x lens is equipped with a correction ring for silicon substrate thickness. The laser diode has two operation modes -- a) low power (200 mW) pulse mode, and b) high power (800 mW) impulse mode. The high power impulse mode can be used for laser fault injection. The laser power can be adjusted from 2$\%$ to 100$\%$ in 0.5$\%$ steps. 

Photon emission analysis~\cite{rahman2019backside} can be used to localize the implemented locked circuitry in the DUT. Thereafter, the DUT is placed under the laser source for LFI. A trigger signal is fed to the PHEMOS-1000 to synchronize the LFI with the DUT operation. Once the device reaches a stable state after power-on, the laser is triggered on the target key registers. After the fault injection, we need to guarantee that the device is still functioning as expected and has not entered into a completely dysfunctional state. The laser triggering timing can be checked by a digital oscilloscope for greater precision. 


\subsection{Fault Count Comparison} \label{subsec:fault-count-comp}

The differential attack methodology (DFA) introduced in Section \ref{sec:background} and in~\cite{jain2020atpg} requires $\mathcal{K}-1$ number of constraints per test pattern. The total number of faults that need to be injected to determine one key bit is $2\mathcal{K}-1$, as $C_A$ and $C_F$ require $\mathcal{K}-1$ and $\mathcal{K}$ faults, respectively. The total number of faults required to decipher $\mathcal{K}$ key bits is $(2\mathcal{K}-1)\cdot\mathcal{K} = 2\mathcal{K}^2-\mathcal{K}$. Compared to DFA, AFIA only requires injecting faults to key registers if these key bits are interdependent, where the propagation of one key is dependent on others. 

Table~\ref{tab:fault-compare} shows the number of faults to be injected for both the DFA (Algorithm~\ref{alg:TP-generation}) and AFIA (Algorithm~\ref{alg:afia}). To demonstrate the feasibility of the fault injection attack on logic locking, we computed the number of faults after generating test patterns using constrained ATPG using the Synopsys TetraMAX tool~\cite{SynopsysTetraMAX}. Note that the successful generation of test patterns using constrained ATPG guarantees the successful attack on locking. We choose benchmark circuits with  random logic locking (added `-RL' after the benchmark name) and strong logic locking (added `-SL') from TrustHub~\cite{salmani2018trust}, SFLL-hd (added `SFLL-hd'), SFLL-flex (added `SFLL-flex'), and SFLL-rem (added `SFLL\_rem') benchmarks from \cite{yasin2017provably}, and GitHub~\cite{github-sfllrem}. Column 2 represents the secret key size, whereas Columns 3 and 4 represent the number of faults to determine the entire key for DFA and AFIA, respectively. Data in Column 4 is collected under \textit{sa1} fault in test pattern generation (Algorithm \ref{alg:afia}). Finally, Column 5 shows the average number of faults to evaluate one key bit under AFIA. For example, with locked benchmark c432-RN320, the number of faults required for DFA is 2016, whereas AFIA requires only 48 faults to extract the 32 key bits, leading to 1.5 faults per key bit. For c1355-SL1280, the number of faults increased significantly to 32,640 for DFA. AFIA only requires 1,419 faults to determine the 128 key bits, or 11.09 faults per key bit.

\begin{table}[t]
\caption{Comparison of Number of Injected Faults} \vspace{-10px}
\label{tab:fault-compare} 
\begin{center}
\begin{tabular}{|M{1.2in}|M{0.3in}|M{0.4in}|M{0.3in}|M{0.3in}|}
\hline
{\multirow{2}{*}{\textbf{Locked Benchmark}}} & {\textbf{Key Size}} & \textbf{DFA} & \multicolumn{2}{c|}{\textbf{AFIA}} \\ \cline{3-5} 
 &   \textbf{$\bm{(\mathcal{K})}$} & \bm{$F_T$} & \bm{$F_{T}$} & \bm{${F_{T}}/{\mathcal{K}}$} \\ \hline
c432-RN320 & 32 & 2016 & 48 & 1.5 \\ \hline
c432-RN640 & 64 & 8128 & 165 & 2.58 \\ \hline
c432-RN1280 & 128 & 32640 & 1085 & 8.48 \\ \hline
c2670-RN1280  & 128 & 32640 & 520 & 4.06 \\ \hline
c3540-RN1280 & 128 & 32640 & 268 & 2.09 \\ \hline
c5315-RN1280 & 128 & 32640 & 282 & 2.20 \\ \hline
c6288-RN1280 & 128 & 32640 & 268 & 2.09 \\ \hline
c7552-RN1280 & 128 & 32640 & 334 & 2.61 \\ \hline
c1355-SL1280 & 128 & 32640 & 1419 & 11.09  \\ \hline
c1908-SL1280 & 128 & 32640 & 654 & 5.11  \\ \hline
c5315-SL1280 & 128 & 32640 & 3469 & 27.10 \\ \hline
c6288-SL1280 & 128 & 32640 & 368 & 2.88 \\ \hline
c7552-SL1280 & 128 & 32640 & 188 & 1.47 \\ \hline
b14\_C\_k8\_SFLL-hd  & 8 & 120 & 28 & 3.5 \\ \hline
b14\_C\_k16\_SFLL-flex  & 16 & 496 & 120 & 7.5 \\ \hline
b14\_C\_k32\_SFLL-flex  & 32 & 2016 & 496 & 15.5 \\ \hline
b14\_C\_k64\_SFLL-flex  & 64 & 8128 & 2016 & 31.5 \\ \hline
b14\_C\_k128\_SFLL-flex  & 128 & 32640 & 8128 & 63.5 \\ \hline
c432\_k8\_SFLL-hd  & 8 & 120 & 28 & 3.5 \\ \hline
c432\_k16\_SFLL-flex  & 16 & 496 & 120 & 7.5 \\ \hline
c432\_k32\_SFLL-flex  & 32 & 2016 & 496 & 15.5 \\ \hline
c880\_k8\_SFLL-hd  & 8 & 120 & 28 & 3.5 \\ \hline
c880\_k16\_SFLL-flex  & 16 & 496 & 120 & 7.5 \\ \hline
c880\_k32\_SFLL-flex  & 32 & 2016 & 496 & 15.5 \\ \hline
SFLL\_rem\_k128~\cite{github-sfllrem} & 128 & 32640 & 8128 & 63.5 \\\hline
\end{tabular}
\end{center} 
\end{table}

Based on Theorem \ref{th:average-worst}, if keys are uniformly distributed among logic cones, the number of fault injections for AFIA is linear with respect to key size, $\bm{O}(a\mathcal{K})=\bm{O}(\mathcal{K})$, with variable $a$ indicating the average key size per logic cone. If having the same key size, an RLL circuit with more logic cones, or a smaller $a$, (provided that the size of all logic cones are about the same), should, generally, has fewer fault injections than one with fewer logic cones. This is equivalent to having fewer injected faults in an RLL-based circuit that contains more output than the ones without (see definition of the number of logic cones in \ref{subsec:coneA}). Benchmark c432-RN1280 has a larger $a$ than other 128-bit RLL circuits, for c432 has only seven outputs, while c2670 has 140 outputs, c3540 has 22, c5315 has 123, c6288 has 32, c7552 has 108 outputs respectively. (Note, not all logic cones will have keys inside, but the circuit with more output usually has more key-embedded cones than those with fewer outputs.) This is the reason that c432-RN1280 requires considerably more fault injections in total, 1085, than other locked netlist with same key size, where c2670-RN1280 needs 520 faults, c3540-RN1280 has 268, c5315-RN1280 has 282, c6288-RN1280 has 268, c7552-RN1280 has 334, see Table \ref{tab:fault-compare}. 

RLL randomly picks a location in the original unlocked circuit for key gate insertion, while SLL produces more blocking keys. In terms of theoretical complexity analysis, as long as the key gates in RLL locked circuit are distributed uniformly, the number of fault injections for SLL should be larger than RLL, under the same original unlocked benchmark and the same key size, \eg, c5315-RN1280 and c5315-SL1280, c6288-RN1280 and c6288-SL1280. For SFLL-hd and SFLL-flex, each locked circuit has a perturbation unit and a restoration unit. All keys reside in the functionality restoration unit, where every key passes through the output of the restoration subcircuit to reach the primary output~\cite{sirone2020functional, alrahis2021gnnunlock}. Because of this restoration unit, all key bits are interdependent. Hence, all SFLL-flex and SFLL-hd circuits belong to the worst-case scenario as in Theorem \ref{th:worst-case}, in which the number of injected faults is $\frac{\mathcal{K}\cdot(\mathcal{K}-1)}{2}$.
We also evaluated our proposed attack on the latest SFLL variant, SFLL-rem~\cite{sengupta2018atpg,sengupta2020truly}. Although SFLL-rem does not have the added perturb unit, the keys are present in the restoration unit only, and our attack can still break it.

\section{Future Work} 
Although AFIA targets combinational logic circuits or sequential ones with scan-chain access, it can be extended to other clock-based and timing-based locking techniques that target output corruptibility in a different clock cycle~\cite{rahman2022oclock,sweeney2020latch, azar2021data}. These techniques require multiple clock cycles (typically two) to capture the key to a storage element and thus observe its effect on the circuit behavior (i.e., output corruptibility). Fortunately, the same fault injection-based attack proposed in this paper can be applied to these locking techniques as well. We, however, need to consider transition delay faults (\textit{TDFs}) or path delay faults (\textit{PDFs}) instead of stuck faults to propagate the effect of the targeted key on the output. The same Algorithm 2 can be applied to generate patterns to launch the attack. Note that the \textit{TDFs} and \textit{PDFs} require multiple captures (typically 2). By controlling the fault injection in a precise timing range, it is possible to observe the key through launch on shift (\textit{LOS}) and launch on capture (\textit{LOC}) schemes~\cite{bushnell2004essentials, savir1994broad}.

\vspace{-10px}
\section{Conclusion}\label{sec:conclusion}
This paper presents AFIA, a novel stuck-at fault-based fault injection attack that undermines the security of any logic locking technique. AFIA utilizes cone analysis to analyze the dependency of keys. Faults are injected only at the interdependent key bits, which is a significant improvement from the previously published attack DFA~\cite{jain2020atpg}, dropping the total number of faults to the linear multiple of key size. With the automatic test pattern generation (ATPG) tool, we constructed a pattern set, which is used to apply to an unlocked chip. Each pattern is sufficient to determine a one-bit key. All key bits are derived by comparing collected responses from fault injections and the predicted response from test pattern generation. We performed laser fault injections on Kintex-7 FPGA with various locked benchmark circuits and state-of-the-art locking techniques, and our results have demonstrated the effectiveness of the proposed AFIA scheme. Our future work will focus on developing a locking technique to prevent AFIA.

\vspace{5px}
\noindent \textbf{Funding:} This work was supported by the National Science Foundation under Grant Number CNS-1755733. 

\vspace{5px}
\noindent \textbf{Data Availability:} The authors declare that the data supporting the findings of this study are available within the article.

\vspace{-10px}
\section*{Declarations}
\vspace{-10px}
\noindent\textbf{Conflict of Interest/Competing Interest:} The authors have no conflicts of interest to declare that are relevant to the content of this article.
\bibliographystyle{spmpsci}


\vspace{10px}
\noindent \textbf{Yadi Zhong} is currently pursuing her Ph.D. in Computer Engineering from the Department of Electrical and Computer Engineering, Auburn University, AL, USA. She received her B.E. degree from the same university in 2020. Her research interest are logic locking, fault injection and hardware security, and post-quantum cryptography. She received Auburn University Presidential Graduate Research Fellowships in 2020.

\vspace{10px}
\noindent\textbf{Ayush Jain} received his M.S. Degree from the Department of Electrical and Computer Engineering, Auburn University, AL, USA in 2020. He is currently working as SoC Design Engineer at Intel Corporation. He received his B.Tech degree from the Electrical Engineering Department, Pandit Deendayal Petroleum University, Gujarat, India, in 2018. His current research interests include hardware security, VLSI design, and testing.

\vspace{10px}
\noindent\textbf{M Tanjidur Rahman} received his Ph.D. degree in electrical and computer engineering in 2021 from University of Florida. He obtained his BS (with honors) and MS in electrical and electronic engineering from Bangladesh University of Engineering and Technology (BUET) in 2012, and 2014, respectively. His research interests include hardware security and trust, physical assurance, configurable security architecture, and reliable VLSI design. Dr. Rahman has authored 7 technical journal and 8 conference papers. He has also published one book and one book chapter on physical assurance and chip backside security assessment. He has two patent applications under review. 

\vspace{10px}
\noindent\textbf{Navid Asadizanjani} received the Ph.D. degree in Mechanical Engineering from University of Connecticut, Storrs, CT, USA, in 2014. He is currently an Assistant Professor with the Electrical and Computer Engineering Department at University of Florida, Gainesville, FL, USA. His current research interest is primary on “Physical Attacks and Inspection of Electronics”. This includes wide range of products from electronic systems to devices. He is involved with counterfeit detection and prevention, system and chip level reverse engineering, Anti reverse engineering, etc. Dr. Asadizanjani has received and nominated for several best paper awards from International Symposium on Hardware Oriented Security and Trust (HOST) and International Symposium on Flexible Automation (ISFA). He was also winner of D.E. Crow Innovation award from University of Connecticut. He is currently the program chair of the PAINE conference and is serving on the technical program committee of several top conferences including International Symposium of Testing and Failure Analysis (ISTFA) and IEEE Computing and Communication Workshop and Conference (CCWC).

\vspace{10px}
\noindent\textbf{Jiafeng (Harvest) Xie} received the M.E. and Ph.D. from Central South University and University of Pittsburgh, in 2010 and 2014, respectively. He is currently an Assistant Professor in the Department of Electrical \& Computer Engineering, Villanova University, Villanova, PA. His research interests include cryptographic engineering, hardware security, post-quantum cryptography, and VLSI implementation of neural network systems. Dr. Xie has served as technical committee member for many reputed conferences such as HOST, ICCD, and ISVLSI. He is also currently serving as Associate Editor for Microelectronics Journal and IEEE Access. He was serving as Associate Editor for IEEE Transactions on Circuits and Systems-II: Express Briefs. He received the IEEE Access Outstanding Associate Editor for the year of 2019. He also received the Best Paper Award from HOST'19.

\vspace{10px}
\noindent\textbf{Ujjwal Guin} received his PhD degree from the Electrical and Computer Engineering Department, University of Connecticut, in 2016. He is currently an Assistant Professor in the Electrical and Computer Engineering Dept. of Auburn University, Auburn, AL, USA. He received his B.E. degree from the Dept. of Electronics and Telecommunication Engineering, Bengal Engineering and Science University, Howrah, India, in 2004 and his M.S. degree from the Dept. of Electrical and Computer Engineering, Temple University, Philadelphia, PA, USA, in 2010. He has developed several on-chip structures and techniques to improve the security, trustworthiness, and reliability of integrated circuits. His current research interests include Hardware Security \& Trust. He has authored several journal articles and refereed conference papers. He serves on the organizing committees of HOST, VTS, and PAINE, and  technical program committees of DAC, HOST, VTS, PAINE, VLSID, GLSVLSI, ISVLSI, and Blockchain. He is an active participant in the SAE International G-19A Test Laboratory Standards Development Committee and G-32 Cyber-Physical Systems Security Committee. He is a member of ACM and a senior member of IEEE.

\end{document}